\documentclass[11pt, a4paper]{article}

\usepackage{parskip,setspace}
\usepackage{amsmath,amssymb,amsthm}
\usepackage[hmargin=25mm,vmargin=25mm]{geometry}
\usepackage{listings}
\usepackage{textcomp}
\usepackage[pdftex]{graphicx}
\usepackage[utf8]{inputenc}
\usepackage[T1]{fontenc}
\usepackage[numbers,sort&compress]{natbib}
\usepackage{longtable}

\newtheorem{Theorem}{Theorem}

\theoremstyle{definition}

\newcommand{\vect}[1]{\mathbf{#1}}
\newcommand{\xvec}{\vect{x}}

\DeclareMathOperator*{\argmin}{\arg\min}

\usepackage[colorlinks=true,linkcolor=blue,urlcolor=blue,citecolor=blue,bookmarks=false]{hyperref}

\usepackage{color}

\usepackage[colorinlistoftodos]{todonotes}

\title{A Hybrid Quantum-Classical Paradigm to Mitigate Embedding Costs in Quantum Annealing}
\author{Alastair A.\ Abbott\thanks{University Grenoble Alpes, CNRS, Grenoble INP, Institut N\'eel, 38000 Grenoble, France}, Cristian S.\ Calude\thanks{Department of Computer Science, University of Auckland, Private Bag 92019, Auckland, New Zealand}, Michael J.\ Dinneen\footnotemark[2] and Richard Hua\footnotemark[2]}

\date{\today}

\begin{document}

\maketitle

\begin{abstract}
	Despite rapid recent progress towards the development of quantum computers capable of providing computational advantages over classical computers, it seems likely that such computers will, initially at least, be required to run in a hybrid quantum-classical regime.
	This realisation has led to interest in hybrid quantum-classical algorithms allowing, for example, quantum computers to solve large problems despite having very limited numbers of qubits.
	Here we propose a hybrid paradigm for quantum annealers with the goal of mitigating a different limitation of such devices: the need to embed problem instances within the (often highly restricted) connectivity graph of the annealer.
	This embedding process can be costly to perform and may destroy any computational speedup. 
	In order to solve many practical problems, it is moreover necessary to perform many, often related, such embeddings.
	We will show how, for such problems, a raw speedup that is negated by the embedding time can nonetheless be exploited to give a real speedup.
	As a proof-of-concept example we present an in-depth case study of a simple problem based on the maximum weight independent set problem.
	Although we do not observe a quantum speedup experimentally, the advantage of the hybrid approach is robustly verified, showing how a potential quantum speedup may be exploited and encouraging further efforts to apply the approach to problems of more practical interest.  
\end{abstract}
	
Quantum computation has the potential to revolutionise computer science, and as a consequence has, since its inception, received a great deal of attention from theorists and experimentalists alike.
Although much progress has been made through the concerted efforts of the community, we are still some distance from being able to build sufficiently large-scale universal quantum computers to realise this potential~\cite{Ladd:2010aa,Barends:2016aa}.

More recently, however, significant progress has been made in the development of special-purpose quantum computers.  
This has been driven by the realisation that, by dropping the requirement of being able to efficiently simulate arbitrary computations and relaxing some of the constraints that make large-scale universal quantum computing difficult (e.g., the ability to apply gates to arbitrary pairs of, possibly non-adjacent, qubits), such devices can be more easily engineered and scaled.  
The expectation is that with this approach one may be able to exploit some of the capabilities of quantum computation---even if its full abilities are for now beyond our reach---to obtain lesser, but nevertheless practical, advantages in  practical  applications.  Quantum annealers, which solve particular optimisation problems, exemplify this approach, and significant progress has been made in recent years towards engineering moderately large-scale such devices~\cite{Johnson:2011aa,Boixo:2016aa}.  
This approach has been pursued particularly zealously by D-Wave, who have developed quantum annealers with upwards of 2000 qubits (e.g., the D-Wave 2000Q\texttrademark{} machine~\cite{dwavesys2017}), and are thus of sufficient size to tackle problems for which their performance can meaningfully be compared to classical computational approaches. 

In this paradigm, however, it is non-trivial to compare the performance of quantum solutions to classical ones, since the focus is on obtaining  practical  gains in domains where heuristics tend to be at the core of the best classical approaches.
Indeed, this issue is at the heart of recent debate surrounding the performance of D-Wave machines~\cite{Shin:2014aa,Cho:2014aa}.
In particular, instead of focusing on asymptotic analyses, one must compare the performance of classical and quantum devices
empirically.
But performing benchmarks fairly is difficult, especially when there is often debate as to which classical algorithm should be taken for comparison~\cite{Ronnow:2014aa,King:2015aa,King:2015ab,Heim:2015aa}.
This is further complicated by the crucial realisation that such special-purpose quantum devices are operated in a fundamentally different way to the classical ones with which they are often compared:
typically, they operate in conjunction with a non-trivial pipeline of classical pre- and post-processing whose contribution is far from negligible on the performance of the device, and may even be the difference between obtaining a quantum speed-up or not. 
Note that such pre- and post-processing costs may also arise when generic classical solvers (e.g.~Integer Programming or SAT solvers) are used for optimisation problems, and although such solvers may not be the fastest classical algorithms for a given problem they are nonetheless of much practical interest and, when compared to quantum annealers, this processing pipeline should similarly be taken into account.

In this paper, motivated by the need to take into account the cost of  classical processing in benchmarking quantum annealers, we propose a hybrid quantum-classical approach for developing algorithms that can mitigate the cost of this processing.
In particular,  we focus on D-Wave's quantum annealers, where this processing involves a costly classical ``embedding'' stage that maps an arbitrary problem instance into one compatible with D-Wave's limited connectivity constraints.
This embedding is generally very time-consuming, and experimental studies indicate that its quality can have strong effects on performance~\cite{jobshopschedulingproblem2016}. 
Indeed, hybrid approaches themselves have previously been used to reduce the cost and size of these embeddings~\cite{boostingQA2017}.
We formulate a hybrid approach that can mitigate this cost on problems where many related embeddings must be performed by modifying the problem pipeline to reuse or modify embeddings already performed, thereby allowing any potential advantage to be accessed more directly~\cite{Calude:2015aa}.
A similar type of approach has previously been suggested as a theoretical means to exploiting Grover's algorithm~\cite{Lanzagorta:2005aa}, and differs from recent hybrid approaches for quantum annealing~\cite{aqc_heuristic2011,Tran:2016aa,McClean:2016aa,Chancellor:2016aa,Grass:2016aa} and computing~\cite{Bauer:2016aa,Li17} that instead aim to provide quantum advantages in situations where far fewer qubits are available than would be needed to execute a complete quantum algorithm for the problem in question~\cite{maxindepsetqubo2017,10.3389/fict.2016.00014,bian14}.
Research thus far has focused on using quantum annealing to solve problems for which only a single embedding is required.
The hybrid approach we propose therefore draws attention to the fact problems to which it can be applied---which require many embedding steps---are more promising candidates for observing practical quantum speedups, and hence serves also to help in guiding the search for problems suitable for quantum annealing.

Having outlined this hybrid computing approach, we then present a hybrid algorithm that is based around a D-Wave solution to the maximum-weight independent set (MWIS) problem.
Although the problem this algorithm solves, called the dynamically weighted MWIS problem, perhaps has limited independent interest and represents a rather simple application of our more general approach, it serves as a strong proof-of-concept for more complex algorithms, and we reinforce this by implementing it experimentally on a D-Wave 2X machine~\cite{DWave2X}.
The results of the experiment show a large improvement of the hybrid algorithm over a standard quantum annealing approach, in which the embedding process is naively repeated many times.
We further compare the hybrid algorithm to a standard classical algorithm.
Although we do not observe an overall speedup using the hybrid algorithm, the scaling behaviour of this approach compares favourably to that of the classical algorithm, leaving open the possibility of future speedups for this problem.

The outline of this paper is as follows.
In Section~\ref{sec:dwaveFramework} we present an overview to (D-Wave's approach to) quantum annealing and benchmarking such devices.
In doing so, we are deliberately thorough and pedagogical, since unfair or poor benchmarking has been the source of much misunderstanding regarding quantum speedups, and is crucial to the approach we outline.
In Section~\ref{sec:hybrid} we present, in a general setting, our hybrid paradigm.
In Section~\ref{sec:caseStudy} we provide an illustrative case study, applying our approach to the dynamically weighted maximum-weight independent set problem and compare its performance on a   D-Wave 2X machine to the standard quantum annealing pipeline.
Finally, in Section~\ref{sec:conclusions} we present our conclusions.

\section{D-Wave's quantum annealing framework}
\label{sec:dwaveFramework}

\subsection{Quantum annealing}

Quantum annealing is a finite temperature implementation of adiabatic quantum computing~\cite{Farhi:2000aa}, in which the optimisation problem to be solved is encoded into a Hamiltonian $H_p$ (the quantum operator corresponding to the system's energy) such that the ground state(s) of $H_p$ correspond(s) precisely to the solution(s) to the problem (of which there may be several).
The computer is initially prepared in the ground state of a Hamiltonian $H_i$, which is then slowly evolved into the target Hamiltonian $H_p$.
This computation can be described by the time-dependent Hamiltonian $H(t)=A(t)H_i + B(t)H_p$ for $0\le t \le T$, where $A(0)=B(T)=1$ and $A(T)=B(0)=0$.
$T$ is called the \emph{annealing time} and the functions $A$ and $B$ determine the \emph{annealing schedule} (for details on D-Wave's schedule, see Refs.~\cite{Johnson:2011aa} and~\cite{King:2016aa}).

If the computation is performed sufficiently slowly, the Adiabatic Theorem guarantees that the system will remain in a ground state of $H_p$ throughout the computation and the final state will thus correspond to an optimal solution to the problem at hand~\cite{Farhi:2000aa}.
In the ideal adiabatic limit, the time required for such a computation scales as the inverse-square of the minimum spectral-gap\footnote{Determining the minimum spectral gap, and thus the time required for computation, is unfortunately itself a computationally difficult problem~\cite{King:2014aa}.} (i.e., the minimum difference between the ground and first excited states of $H(t)$).
However, in the finite temperature regime of quantum annealing, a trade-off must be found between evolving the system sufficiently slowly and avoiding the perturbing affect of the environment.
As a consequence, the final state is only a correct solution with a certain probability, and the (hence probabilistic) computation must be repeated many times to obtain the desired solution (or a sufficiently close approximation thereof)~\cite{Johnson:2011aa,McGeoch2014}.

\subsection{Quadratic unconstrained Boolean optimisation}

Although the adiabatic computational model is quantum universal~\cite{Mizel:2007aa}, the recent success of quantum annealing has come about by focusing on implementing specific types of Hamiltonians that are simpler to engineer and control, despite the fact they might not be capable of efficiently simulating arbitrary quantum circuits.
In particular, D-Wave's devices can be modelled by a two-dimensional Ising spin glass Hamiltonian, and it is thus capable of solving the \emph{Ising spin minimisation problem}, a well-known NP-hard optimisation problem~\cite{Istrail:2000aa,Johnson:2011aa}.
This problem is equivalent, via a simple mapping of spin values ($\pm 1$) to bits (0 or 1), to the \emph{Quadratic Unconstrained
Boolean Optimisation (QUBO) problem}~\cite{Choi:2008aa}. In this paper we will use this formulation, as it will allow us to represent
in detail a little more compactly the algorithms.

The QUBO problem is the task of finding the input $\vect{x}^*$ that minimises a quadratic objective function of the form $f(\vect{x}) = \vect{x}^{T}Q\vect{x}$, where $\vect{x}=(x_1,\dots,x_n)$ is a vector of $n$ binary variables and $Q$ is an upper-triangular $n \times n$ matrix of real numbers:
\begin{equation}
\label{eqn:qubo}
\vect{x}^* = \argmin_{\xvec}\vect{x}^{T}Q\vect{x}  = \argmin_{\xvec} \sum_{i\leq j} x_i Q_{(i,j)} x_j, \text{ where } x_i \in \{0,1\}.
\end{equation}
Note that arbitrary quadratic objective functions $g$ can be converted to this form.
Since $x_i^2=x_i$ for $x_i=0$ or $1$, linear terms of $g$ can be encoded as the diagonal entries of a $Q$ for $f$.
Furthermore, any constant terms in $g$ can be ignored since they do not affect the objective minimisation with respect to $\xvec$.

In the quantum annealing model of the QUBO problem, each $x_i$ corresponds to a qubit while $Q$ defines the problem Hamiltonian $H_p$.
Specifically, the non-zero off-diagonal terms $Q_{(i,j)}$, $i<j$, correspond to couplings between qubits $x_i$ and $x_j$, while the diagonal terms $Q_{(i,i)}$ are related to the local field applied to each qubit.
For a given QUBO problem $Q$, these couplings may be conveniently represented as a graph $G_L=(V_L,E_L)$ representing the interaction between qubits, where $V_L=\{1,\dots,n\}$ is the set of qubits and
$E_L = \{\{i,j\} \mid Q_{(i,j)}\neq 0,\ i < j\}$ are the edges representing the couplings between qubits.
For reasons that will soon be apparent, we will refer to such a graph for a given QUBO problem as the \emph{logical graph}, and the set of qubits the QUBO problem is represented over the \emph{logical qubits}.

\subsection{Hardware constraints and embeddings}\label{sec:Chimera}

In practice, it remains exceedingly difficult to control interactions between qubits that are not physically near to one another, and as a result it is not possible to directly implement directly any instance of the QUBO problem: this would require directly coupling arbitrary pairs of qubits, which is currently infeasible. 
Instead, the couplings possible on a quantum annealer are specified by a graph $G_P=(V_P,E_P)$, where $V_P$ is the set of qubits on the device, and an edge $\{i,j\}\in E_P$ signifies that qubits $i$ and $j$ can be physically coupled.
The graph $G_P$ is called the \emph{physical graph}, and the qubits $V_P$ are the \emph{physical qubits}~\cite{Choi:2008aa,Lechner:2015aa}.

The physical graphs implemented on D-Wave's existing devices are \emph{Chimera graphs} $\chi_{k}$, which are $k\times k$ grids of $K_{4,4}$
graphs, with connections between adjacent ``blocks'' as shown in Figure~\ref{fig:ChimeraGraph}.%
\footnote{It is possible to define a more general family of $n\times m \times L$ Chimera graphs that are
$n\times m$ grids of $K_{L,L}$ graphs, as in Ref.~\cite{dwavebroadcast2016}. However, all devices to date have been square grids of $K_{4,4}$ graphs and, so that, in order to talk more precisely about scaling behaviour, we adopt the convention of fixing $L=4$ and $n=m$~\cite{King:2014aa,King:2015aa}. This is further justified by noting that increasing $L$ involves increasing the \emph{density} of qubit couplings, which is technically much more difficult than increasing the grid size.} 
Specifically, each qubit is coupled with 4 other qubits in the same $K_{4,4}$ block and 2 qubits in adjacent blocks (except for qubits in blocks on the edge of the grid, which are coupled to a single other block).
See Ref.~\cite{Saket:2013aa} for a more formal definition of the Chimera graph structure.

\begin{figure}[t]
\begin{center}
\includegraphics[scale=0.5]{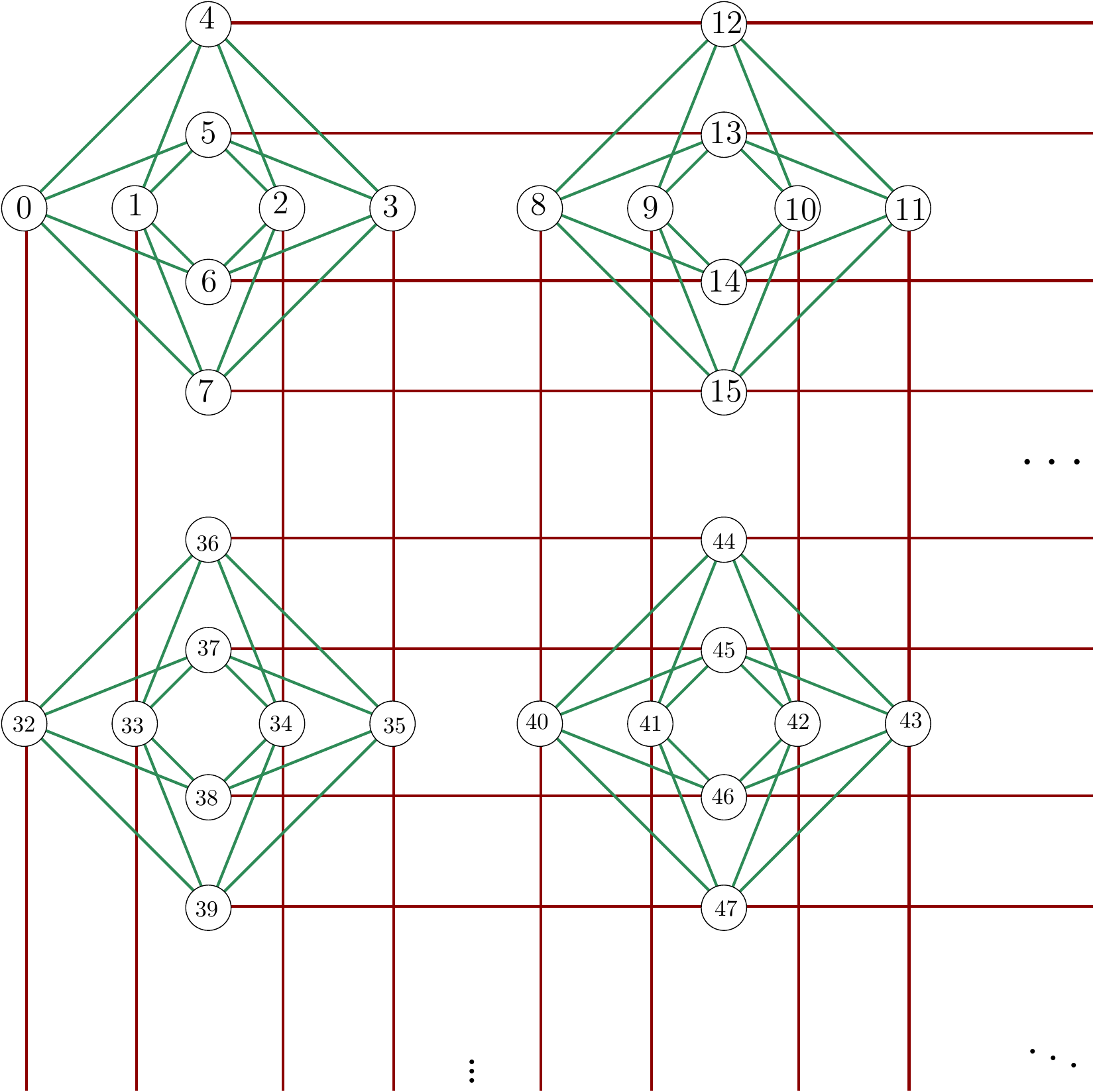}
\caption{A portion of a Chimera graph, showing four $K_{4,4}$ blocks. In general, the graph $\chi_{k}$ consists of a $k\times k$ grid of such blocks, with connections between adjacent blocks as shown.}
\label{fig:ChimeraGraph}
\end{center}
\end{figure}

The Chimera graph is, crucially, relatively sparse and quasi-two-dimensional, with qubits separated by paths of length no longer than $2k$.
Although the specific choice of hardware graph is an engineering decision and may conceivably be changed in future devices, any
alternative physical graph is likely to have similar properties since the tradeoff between connectivity and practicability is a core
feature (and intrinsic limitation) of the current approach to quantum annealing~\cite{Choi:2011aa,Lechner:2015aa}.\footnote{Indeed, D-Wave recently announced that future devices will have a different physical graph, the ``Pegasus'' graph~\cite{pegasus2019}.}
It is therefore essential to take into account these limitations of the hardware graph in any approach to solving problems with quantum annealers.

Since the logical graph $G_L$ for a QUBO problem instance $Q$ will not, in general, be a subgraph of the physical graph $G_P=\chi_k$, the problem instance on $G_L$ must be mapped to an equivalent one on $G_P$.
This process involves two steps: first, $G_L$ must be \emph{embedded} in $G_P$, and secondly the weights of the QUBO problem (i.e., the non-zero entries in $Q$) must be adjusted so that valid solutions on $G_P$ are mapped to valid solutions on $G_L$.

The embedding stage amounts to finding a \emph{minor embedding} of $G_L=(V_L,E_L)$ into $G_P=(V_P,E_P)$~\cite{Barahona:1982aa,Choi:2008aa}, i.e., an embedding function $f:V_L\to 2^{V_P}$ such that
\begin{enumerate}
\item the sets of vertices $\{ f(v) \mid v \in V_L\}$ are disjoint;
\item for all $v \in V_{L}$, there is a subset of edges $E' \subseteq E_{P}$ such that $G' = (f(v), E')$ is connected;
\item if $\{u,v\} \in E_{L}$, then there exist $u',v' \in V_{P}$ such that $u' \in f(u)$, $v' \in f(v)$ and $\{u',v'\}$ is an
edge in $E_{P}$.
\end{enumerate}
Typically, this involves mapping each logical qubit to ``chains'' or ``blocks'' of physical qubits.
In general, a QUBO instance using $n$ logical qubits will require up to $O(n^2)$ physical qubits since the smallest Chimera graph in which the complete graph $K_{4k}$ can be embedded in is $\chi_{k}$, requiring $4k(k+1)$ physical qubits~\cite{Choi:2011aa,King:2014aa}.
The embedding thus already entails, in general, a quadratic increase in problem size which needs to be taken into account when benchmarking quantum annealers.

The problem of finding a minor embedding is itself computationally difficult~\cite{Choi:2008aa}.
Of course, if one has sufficiently many physical qubits to embed $K_n$ then any $n$-qubit logical graph can trivially be embedded into the physical graph.
However, this trivial embedding is generally rather wasteful since qubits are precious resources as the practical limits of quantum annealing are still constantly being pushed. 
Perhaps more importantly, as more physical qubits are required the amount of time needed to find a (sufficiently good) solution increases,
so even when such a naive embedding exists there may be a significant advantage in looking for smaller embeddings (the feasibility of a problem may even depend on it).
The embedding process may thus, in light of its computational difficulty, contribute significantly to the time required to solve a problem in practice. 
Currently, the standard approach to finding such an embedding is to use heuristic algorithms (see, e.g., Ref.~\cite{Cai:2014}).

The second stage, which ensures that the validity of solutions is preserved, involves deciding on how to share the weights associated with each logical qubit $v$ between the physical qubits $f(v)$ it it is mapped to.
Since the weights must all fall within a finite range\footnote{Physically, the quantum annealer requires that the QUBO weights satisfy $|Q_{(i,j)}|\le 1$ for all $i$ and $j$. 
An arbitrary problem specified by $Q$ must thus be scaled to satisfy this constraint.} and there is a limited analogue precision with which the weights can be set, this process can effectively amplify the relative effects of such errors and thus decrease the probability of finding the correct solution~\cite{Choi:2008aa,King:2014aa,Pudenz:2016aa,albash19,PhysRevA.88.062314}.
This stage thus further exemplifies the need to avoid unnecessarily large embeddings, but does not have the same intrinsic computational cost as the embedding process proper.

\subsection{Benchmarking quantum annealers}

Although from a theoretical perspective it is expected that general purpose quantum computers will provide a computational advantage over classical algorithms, there has been much debate over whether or not quantum annealing provides any such speedup in practice~\cite{Ronnow:2014aa,Calude:2015aa,Lechner:2015aa}.
Much of this debate has stemmed from disagreement over what exactly constitutes a quantum ``speedup'' and, indeed, how to determine if there is one~\cite{Ronnow:2014aa}. 
In this paper we will focus primarily on the run-time performance in investigating whether a quantum speedup is present, rather than the (empirically estimated) scaling performance of quantum algorithms.

One of the key points complicating this issue is the fact that, even in the standard circuit model of quantum computation, it is not generally believed that an exponential speedup is possible for NP-hard problems such as the QUBO problem~\cite{Aaronson:2010aa}.
Leading quantum algorithms instead typically provide a quadratic or low-order polynomial speedup~\cite{Furer:2008aa}.
In practice, heuristic algorithms are generally used to solve such optimisation problems and the probabilistic nature of quantum annealing means that it is also best viewed in this light~\cite{King:2014aa,Ronnow:2014aa}.
This means that, rather than theoretical algorithmic analysis, empirical measures are essential in benchmarking quantum annealing against classical approaches.

\subsubsection{Measuring the processing time}\label{sec:measure_time}

Good benchmarking will, first of all, need to make use of fair and comprehensive metrics to determine the running time of both classical and quantum algorithms for a problem.
In particular these need to properly take into account not only the ``wall-time'' of different stages of the quantum algorithm, but also its probabilistic nature.
To understand how this can be done, we first need to outline the different stages of the quantum annealing process~\cite{King:2014aa}.
\begin{enumerate}
\item \emph{Programming:} The problem instance is loaded onto the annealing chip (QPU), which takes time $t_{\text{prog}}$.
\item \emph{Annealing:} The quantum annealing process is performed and then the physical qubits are measured to obtain a solution; this takes time\footnote{Note that this is sometimes referred to as the ``wall clock time'' in the literature. For simplicity, we choose to englobe all times associated with an annealing cycle (e.g.\ readout and inter-sample thermalisation times) along with the annealing time \emph{per se} into $t_a$.} $t_a$.
\item \emph{Repetition:} Step 2 is performed $k$ times to obtain $k$ potential solutions.
\end{enumerate}
The \emph{quantum processing time} (QPT) is thus $$t_{\text{proc}}=t_{\text{prog}}+k\,t_a.$$
For any given run of a quantum annealer, there is a non-zero probability of obtaining a correct solution to the problem at hand,  which depends on both the annealing time $t_a$ and the number of repetitions $k$.
Moreover, for any specific problem instance, the optimal values of these parameters are not known \emph{a priori}, so the performance of a quantum annealing algorithm will be determined by the optimal values of these parameters for the hardest problems of a given size $n$~\cite{Ronnow:2014aa}.
On   D-Wave 2X (and earlier) devices, however, the minimal annealing time of $20\mu\text{s}$ has repeatedly been found to be longer than the optimal time~\cite{Hen:2015aa,King:2014aa,Ronnow:2014aa,Venturelli:2015aa}.

A relatively fair and robust way to measure the quantum processing time is  the ``time to solution'' (TTS) metric~\cite{Boixo:2014aa,Ronnow:2014aa}, which is based on the expected number of repetitions needed to obtain a valid solution with probability $p$ (one typically takes $p=0.99$).\footnote{It is possible to generalise the TTS method to a time-to-target (TTT) method~\cite{King:2015aa}, where one is interested in the expected time to obtain a solution that is sufficiently good with respect to some (perhaps problem-dependent) measure.
Although this approach is likely to be very useful in benchmarking larger  practical  problems, we focus on the TTS approach here (which can be seen as a specific case of TTT).\label{fn:TTT}}
If the probability per annealing sample of obtaining a solution is $s$ (which can be estimated empirically), then this is calculated as 
\begin{equation}\label{eq:k99}
	k_{p}=\frac{\log(1-p)}{\log(1-s)}\raisebox{0.7mm}{,}
\end{equation}
and the quantum processing time is thus calculated with this $k$ as $t_{\text{proc}}=t_{\text{prog}}+k_{p}\,t_a$.
Throughout the rest of the paper we will fix $p=0.99$ as is typically done, and thus consider $k_{99}$.

In practice, unfortunately, even for moderate problem sizes, quantum annealing (and, indeed, classical annealing) simply does not find a correct solution to many problem instances~\cite{Boixo:2014aa,Denchev:2016aa,Ronnow:2014aa}.
Thus, although no worst case running time for such problems can be calculated, it is often instructive to look at the QPT for restricted classes of problems of particular interest or of limited difficulty.
In particular, several authors have applied this to difficulty ``quantiles'', calculating the QPT for, e.g., the 75\% of problems that can be solved the quickest. 
Investigating how the QPT scales with problem difficulty in this way permits some comparison with classical algorithms where it would otherwise be difficult or even impossible~\cite{Boixo:2014aa,Ronnow:2014aa}.

Existing investigations have primarily focused on comparing directly the QPT with the processing time of a classical algorithm in order to look for what we call a ``raw quantum speedup''.
However, it is essential to realise that the time used by the QPU and measured by the QPT refers only to a subset of the processing required to solve a given problem instance using a quantum annealer.
Specifically, a complete quantum algorithm for a problem instance $P$ involves, as a minimum requirement, the following steps:
\begin{enumerate}
\item \emph{Conversion:} The problem instance $P$ must be converted into a QUBO instance $Q(P)$, typically via a polynomial-time reduction taking time $t_{\text{conv}}$.
\item \emph{Embedding:} The QUBO problem $Q(P)$ must be embedded into the Chimera hardware graph taking time $t_{\text{embed}}$.
\item \emph{Pre-processing:} The embedded problem is pre-processed, which involves calculating (appropriately scaled) weights for the embedded QUBO problem, taking time $t_{\text{pre}}$.
\item \emph{Quantum processing:} The annealing process is performed on the QPU, taking time $t_{\text{proc}}$.
\item \emph{Post-processing:} The samples are post-processed to choose the best candidate solution, check its validity, and perform any other post-processing methods to improve the solution quality\footnote{On D-Wave's annealer, for example, a local search may optimally be performed to improve the solution quality. The $k$ repetitions that are performed in the quantum processing step are broken into fixed ``batches'' of $k/b$ samples (where $b$ depends on the problem but not on $k$) and batches are post-processed in parallel with the annealing of the following one; this justifies the consideration of this post-processing as contributing towards the constant overhead $t_\text{post}$, as only the post-processing of the final batch contributes to $T_Q$. Note that such post-processing already constitutes a form of hybrid quantum-classical approach.}~\cite{King:2014aa,Pudenz:2016aa} taking time $t_\text{post}$. The QUBO solution must finally be converted back to a solution for the original problem $P$.
\end{enumerate}
The total processing time is thus\footnote{As a convention, we will use lower case letters $t$ for the timings of subtasks, and upper case $T$'s to denote overall times of computation.} 
\begin{equation}\label{eqn:Tq}
T_Q=t_{\text{conv}}+t_{\text{embed}}+t_\text{pre}+t_{\text{proc}}+t_\text{post}.
\end{equation}

The realisation that these other steps must be included in the analysis is emphasised by the fact that in   practical  problems the embedding time often dominates the time used by the annealer itself.
Previous investigations have largely avoided this by focusing on artificial problems ``planted'' in the Chimera graph so that no embedding is necessary~\cite{Boixo:2014aa,Denchev:2016aa,Hen:2015aa,Ronnow:2014aa,Shin:2014aa}. 
Although finding a raw speedup in such situations is clearly a necessary condition for a quantum speedup, it does not guarantee that any corresponding speedup will carry over into practical problems.

It is therefore the time $T_Q$ which should be used in a fair comparison with classical algorithms.
Note that this still makes use of the TTS approach discussed above, except one must now take into account the tradeoff between the quality of an embedding and the time spent finding it in order to determine the optimal annealing parameters.

\subsubsection{Comparing classical and quantum algorithms}
\label{sec:comparisonToClassical}

To properly benchmark quantum annealing against classical algorithms it is necessary not only to have fair measures of the cost of obtaining a solution, but one must also compare fairly the quantum annealer to a suitable classical algorithm.

Ideally, the performance of a quantum annealer should be compared against the best classical algorithm for the problem being solved.
In practice, such an algorithm is rarely, if ever, known, especially for problems where heuristics dominate, and certain algorithms may perform better on certain subsets of problems.
The best one can do in practice, then, is to look for a ``potential quantum speedup''~\cite{Ronnow:2014aa} by comparing against the best available classical algorithm for the problem at hand. 

Often, however, quantum annealers are also tested against specific classical algorithms of interest; a speedup in such benchmarking has been termed a ``limited speedup'' in Ref.~\cite{Ronnow:2014aa}.
Such studies are important since a limited speedup is, of course, a necessary condition for a real quantum speedup to be present.
This type of benchmarking has often been used, e.g., to compare quantum annealing to simulated annealing or simulated quantum annealing~\cite{Boixo:2014aa,Denchev:2016aa,Shin:2014aa,PhysRevApplied.6.054016}, and such comparisons have the extra benefit of comparing similar use cases---i.e., generic optimisation solvers rather than algorithms tailored to a particular problem and which might require significant development time.
Nonetheless, care should be made in interpreting results when benchmarking in this way, since much of the controversy regarding potential speedup with quantum annealing has arisen when ``limited'' speedups are claimed to have more general relevance.

Finally, it is important to make sure the performance measures for both quantum and classical algorithms are compatible.
That is, the classical processing time $T_C$ should be calculated using a TTS metric as for $T_Q$ (if the classical algorithm is deterministic, this simply reduces to the computation time), and should include all aspects of the classical computation, including pre- and post-processing and reading input.
Note that by including the cost of embedding in the quantum and classical processing times, we make sure that what we  calculate is a function of the problem size $n$ and not the number of physical qubits.

\section{Hybrid quantum-classical computing}
\label{sec:hybrid}

As we discussed in the previous section, most of the effort in determining whether or not quantum annealing can, in practice, provide a computational speedup has focused purely on determining the existence of a \emph{raw quantum speedup}, which does not take into account the associated classical  processing that is inseparable from a quantum annealer.
Such a raw speedup is certainly a necessary condition for practical quantum computational gains, and its study is therefore well justified. However, even if there is a raw speedup there are many reasons why this might not translate into a \emph{practical} quantum speedup.

A practical speedup is possible for a problem if we are able to give a quantum algorithm such that $T_Q<T_C$, where (we recall) $T_C$ is the classical processing time  for the best available classical algorithm for the problem.
From the definition of $T_Q$ in \eqref{eqn:Tq}, it is clear than, even if $t_{\text{proc}}<T_C$, the conversion, embedding and pre/post-processing may provide obstructions to obtaining a practical speedup.
In practical terms, the pre- and post-processing tend to add relatively minor (or controllable) overheads, but the conversion and embedding costs pose more fundamental problems.

The conversion stage can be problematic for two reasons.
First, if the conversion is slow, $t_{\text{conv}}$ may be sufficiently large to negate any speedup.
However, asymptotically $t_{\text{conv}}$ should be polynomial in the problem size $n$, and, in practice for problems suitable for annealing, $t_{\text{conv}}$ seems to be relatively small compared to $t_{\text{proc}}$ and thus has negligible impact on the ability to find an absolute speedup.

More importantly though is the fact that  the QUBO instance resulting from the conversion may be significantly larger than the original problem instance, and thus it can be too large to solve with current quantum annealers.
For example, Ref.~\cite{dwavebroadcast2016} studies the QUBO formulation of the well-known Broadcast Time Problem obtained through a reduction from Integer Programming. 
For instances of this problem on graphs with less than $20$ vertices, the corresponding QUBO formulation required up to $1000$ binary variables (and thus logical qubits) which, especially once the problem is embedded in the physical graph, is beyond the reach of current quantum annealing hardware.

The computational cost of embedding the QUBO instance in the hardware graph is, in absolute terms, even more of an obstruction to successful applications of quantum annealing in its current state.
As mentioned earlier, when using standard heuristic algorithms the embedding time $t_{\text{embed}}$ is generally (at best) comparable to $t_{\text{proc}}$ (and, indeed, $T_C$) and often much longer.
Like the issues associated with the conversion, if sufficiently many qubits are available (i.e., quadratic in the QUBO problem size) and can reliably be annealed, then this embedding can be done quickly and this problem could be neglected.
However, this is certainly not the current situation, and ways to mitigate the dominant effect of $t_{\text{embed}}$ will be needed if quantum annealing is to be successfully applied in its current state or imminent future.

These difficulties in turning a raw quantum speedup into a practical advantage have led to significant interest in ``hybrid classical-quantum'' approaches (also called ``quassical'' computations by Allen, see Ref.~\cite{Calude:2015aa}):  hopefully, by  combining quantum annealing with classical algorithms may allow otherwise inaccessible speedups to be exploited.%
\footnote{We note that hybrid approaches have been also proposed (explicitly and implicitly) in other models of quantum computation too. For example, measurement based computation can be seen a hybrid approach: one starts with a quantum state and performs iterative rounds of quantum measurements and classical computations determining future measurements~\cite{measurement_basedqc2016,Briegel09}.}
Several such hybrid approaches have aimed to overcome the resource limitation arising from the fact that   practical  problems typically require more qubits than are available on existing devices (as a result of the expansion in number of variables during the conversion stage discussed above)~\cite{McClean:2016aa,Tran:2016aa}.
Such proposals instead provide algorithms that utilise quantum annealing on smaller, more manageable subproblems before combining the results classically into a solution for the larger problem at hand.
Other hybrid approaches have aimed to combine quantum annealing with classical annealing and optimisation techniques, in particular by using quantum annealing to perform local optimisations and classical techniques to guide the global search direction~\cite{Chancellor:2016aa,Grass:2016aa}.
These approaches aim to make the most of both quantum advantages (e.g.\ tunnelling) and classical ones (the ability to read and copy intermediate states).

\subsection{Hybrid computing to mitigate minor-embedding costs}
\label{sec:hybridAlgGen}

Although hybrid approaches have also looked at improving the robustness and quality of embeddings~\cite{Vinci:2015aa}, to the best of our knowledge such approaches have not been used to try and mitigate the cost of performing the embedding itself, which, we recall, is often prohibitive to any speedup.
In this paper we propose a general hybrid approach to tackle precisely this problem.
In particular we aim to show how a raw speedup that is negated by the embedding time (i.e., in particular when $t_{\text{proc}}<T_C$ but $T_Q>T_C$) can nonetheless be exploited to give a practical speedup to certain computational problems.

Our approach is motivated by another hybrid quantum-classical algorithmic proposal which predates the rise of quantum annealing and was introduced with the aim of exploiting Grover's algorithm---the well-known black-box algorithm for quantum unordered database search~\cite{Grover:1996aa}---in practical applications~\cite{Lanzagorta:2005aa}.
The motivation in this case was the realisation that, although Grover's algorithm offers a provable quantum speedup, it applies in rather artificial scenarios: it assumes the existence of an unsorted quantum database, when generally a more practical database design would allow for even better speedups, and in most conceivable practical scenarios a costly pre-processing step is needed to prepare the database which immediately negates the quantum speedup.
The authors showed, however, that some more complex practical problems can be approached by solving a large number of instances of 
unstructured database searches on a single database---precisely the problem that Grover's algorithm is applicable to.
Specifically, they looked at practical problems in computer graphics, such as intersection detection in ray-tracing algorithms.\footnote{Here, one must determine the intersections between large numbers of \emph{a priori} unordered three-dimensional objects, which can be rephrased as a search for an initially unknown number of items in an unordered database.}
The need to run Grover's algorithm many times to solve such problems means that the cost of preparing and pre-processing the database can be averaged out over all the runs, thus allowing the theoretical quantum speedup to be recovered.
An important aspect of the hybrid approach of Ref.~\cite{Lanzagorta:2005aa} is that it is not just an algorithmic paradigm for using a quantum computer, but it is also concerned with determining which problems we should try and use the quantum computer to solve.

Although their hybrid approach applies to  a very different situation than that of quantum annealing, there are some clear similarities between the prohibitive costs of preparing the database for Grover's algorithm, and that of performing the embedding prior to annealing.
We thus suggest adopting an analogous approach of using a quantum annealer to solve more complex problems that require solving sets of related (sub)problems whose potential quantum speedup is hidden behind the cost of the embedding required to solve the (sub)problem.
In particular, it might be easier to observe (and thus take advantage of) a quantum speedup by looking at algorithms that require a large number of calls to a quantum annealer as a subroutine, rather than trying to observe a speedup for solving an individual problem instance on an annealer (e.g., a single instance of an NP-complete problem such as the Independent Set problem via a reduction to a single QUBO instance) as previous attempts to use quantum annealing have done.

The crucial condition for a problem to be amenable to this hybrid approach is that \textit{the repeated calls to the quantum annealer should be made with the same logical graph embedding}, or \textit{permit an efficient method to construct the embedding for one call from the previous ones.}
If this condition is satisfied, the cost of the embedding, $t_{\text{embed}}$, can thus be spread out over the several calls, allowing a raw quantum speedup to be exploited.
There are several conceivable ways such a scenario could naturally occur in realistic algorithmic problems, and we will discuss and analyse an example in detail in the following sections.
Perhaps the most trivial would be that where all (or most) solutions to a highly-degenerate problem are required to be found, rather than simply a single one.
Although such a scenario is clearly suitable for quantum annealing, given its intrinsic ability to randomly sample solutions, there are other, perhaps more subtle, situations where this hybrid approach could be applied.
For example, one may need to solve a large number of instances of a problem, $P_1,\dots,P_m$, where the instances $P_i$ differ in some parameters, but where the embedding is independent of these parameters (e.g., if they are encoded in the weights rather than couplings of the logical graph), or if the logical graphs $G_i$ of each instance $P_i$ differ only slightly and are all subgraphs of a single logical graph $G$ that can be embedded.\footnote{Of course, one would want $G$ to be not much larger than the $G_i$, otherwise the embedding of $G$ is unlikely to allow one to compute \emph{good} embeddings of the $G_i$.}
These examples are certainly not definitive, and other situations suitable for this hybrid approach are bound to be uncovered.

In order to see how this hybrid approach can help exploit a quantum speedup, we will consider the particularly simple case with the following general description of a quantum annealing algorithm based on the hybrid approach described above (a more precise analysis would necessarily depend in part on the algorithm in question): 
some initial classical processing is performed, the embedding of a logical graph into the physical graph is computed, $m$ instances of a QUBO problem are solved on a quantum annealer, with some classical pre- and post-processing occurring between instances, and some final classical computation is optionally performed.
We emphasise, however, that the same approach can be applied to cases where the embedding is reused in a less trivial manner, so long as the cost to go from the embedding of one subproblem to the next is small.
Indeed a key part of the challenge---and future research---is finding suitable problems or criteria for which this is the case; here, our goal is to simply outline the underlying paradigm.

More formally, let us call the overall problem the hybrid algorithm solves $R$, and the $m$ problem instances that must be solved to do so, $P_1,\dots,P_m$.
Recall that the time to solve a single instance $P_i$ on an annealer is $T_Q(P_i)$; as we noted earlier this is, in practical situations, generally dominated by the cost of the embedding and the quantum processing, so $T_Q(P_i)$ can be approximated, for simplicity, as
\begin{align*}
T_Q(P_i) &= t_{\text{conv}}(P_i) + t_{\text{embed}}(P_i) + t_\text{pre}(P_i) + t_{\text{proc}}(P_i) + t_\text{post}(P_i)\notag\\
&\approx t_{\text{embed}}(P_i) + t_{\text{proc}}(P_i),
\end{align*}
where we have explicitly included the dependence on the problem instance.
The hybrid algorithm will thus take time
\begin{align}\label{eqn:hybridBreakdown}
T_H(R) &\approx t_1(R) + t_{\text{embed}}(P_1) + \sum_i\big(t_{\text{proc}}(P_i) + t_2(P_i)\big)\notag\\
&\approx t_1(R) + t_{\text{embed}}(P_1) + \sum_i t_{\text{proc}}(P_i),
\end{align}
where $t_1(R)$ encapsulates any initial and final classical processing associated with combining the solutions $P_i$, and $t_2(P_i)$ is the classical calculation associated with each iteration, which we have assumed to be small compared to $t_{\text{proc}}(P_i)$ since this should simply encompass minor pre- and post-processing between annealing runs, and thus be negligible if the problem is amenable to the hybrid approach.\footnote{More precisely, one expects the annealing time to be exponential in general, and if an exponential amount of classical processing is also required, it seems likely that no speedup will be possible. 
This condition could nonetheless be relaxed to obtain an advantage with the hybrid approach, as long as  a raw speedup is still present
when the annealing and processing times are combined (i.e., $t_{\text{proc}} + t_2$), but negated by the embedding if the annealer is used in the standard, more naive, way; however, we make this assumption to simplify our analysis.}
Note that we have made use of the assumption that $t_\text{embed}(P_1)\approx t_\text{emded}(P_i)$ for $i>1$, which is a criterion on the suitability of a problem for this hybrid approach.

We note immediately that a standard approach with a quantum annealer, performing the embedding for each instance $P_i$, would take time
\begin{equation*}
T_\text{std}(R) \approx t_1(R) + \sum_i\big(t_{\text{embed}}(P_i) + t_{\text{proc}}(P_i)\big).
\end{equation*}
In practice, one could envisage exploiting classical parallelism to reduce the cost of performing the embedding $m$ times by a constant factor.
For simplicity, we will assume that such parallelism is not used, and as long as $m$ is large enough the same conclusions hold.
Thus, since in practice $t_{\text{embed}}$ is comparable to, if not larger, than $t_{\text{proc}}$, we already have \[T_H(R) \ll T_\text{std}(R).\]
Although this conclusion may seem somewhat trivial, it is important in that it shows already how annealing can provide much larger practical gains for such complex algorithmic problems.
Indeed, one may view this result as emphasising the need to choose problems that allow the classical overheads of quantum annealing to be negated.
Thus far, the focus has been on traditional algorithmic problems that are difficult to subdivide; by using quantum annealing in more complex algorithms, this hybrid paradigm allows the real performance of a quantum annealer to be more directly accessed.

More importantly,  \textit{it may allow a raw quantum speedup to be exploited practically.}
To see this, let us consider the case when the best classical algorithm can solve a single instance $P_i$ in time $T_C(P_i)$.\footnote{We emphasise that, since we are interested in practical, not only asymptotic, gains, we can not easily assume that $T_C(P_i)=T_C(P_j)$ for $i\neq j$.}
We are interested, in particular, in the case when a raw quantum speedup (i.e., $t_{\text{proc}}(P_i) < T_C(P_i)$)  is negated by the embedding (i.e., $T_Q(P_i) > T_C(P_i)$).
Although the standard classical approach to solving $R$  is to use the classical algorithm to solve each $P_i$, and would thus take time $t_1(R) + \sum_i T_C(P_i)$, we should not assume this is the best classical approach to solving $R$, and for a fair comparison the hybrid approach should be benchmarked against the best known classical algorithm for $R$.

It is, of course, possible that, for certain problems, a much more efficient classical algorithm exists for solving $R$ when $m$ is large enough (e.g., there might be an efficient way to map solutions of $P_i$ to $P_j$). Such problems are thus not suitable for such a hybrid approach, and so are not of particular interest to us.
Nonetheless, in general a classical algorithm for $R$ may be more intelligent than the standard approach as certain, necessarily minor,\footnote{If not, then again the problem is not suitable for the hybrid approach, as a much more efficient classical algorithmic approach exists.} parts of the computation are likely to be common to solving several $P_i$.
Specifically, we can thus rewrite $T_C(P_i)=t_3(P_i) + t_4(P_i)$, where $t_3$ is small compared to $t_4$.
The best classical algorithm can then, rather generally, be considered to take time 
\begin{align*}
T_C^{\text{best}}(R)&=t_5(R) + t_3(P_1) + \sum_i t_4(P_i)
= t_6(R) + \sum_i t_4(P_i),
\end{align*}
\noindent where $t_6(R)=t_5(R) + t_3(P_1)$ and $t_5(R)$ encapsulates any additional global processing (in analogy to $t_1(R)$ for the quantum approaches).
Crucially, unless the raw quantum speedup is small, we will also have $t_{\text{proc}}(P_i)<t_4(P_i)$.

It is thus easy to see that, 
\begin{quote} \em 
for large enough $m$ (i.e., number of $P_i$ to be solved), we have $T_H(R) < T_C^{\text{best}}(R)$,	
\end{quote}
and thus the raw quantum speedup will translate into an absolute speedup for the hybrid algorithm.
The precise value of $m$ for which such a speedup is obtained will, of course, depend on the problem instances themselves, since the runtime can in practice depend heavily on this. 
Moreover, although $m$ depends on the problem $R$ (it may, for example, scale with the problem size, or be fixed), this analysis shows that there are problems for which this hybrid approach can turn a raw quantum speedup into a practical one.

It is important to reiterate that the quantum (and, if applicable, classical) times should be calculated using the TTS metric for each problem instance in order to correctly take into account the probabilistic nature of the quantum (and, potentially, classical) algorithms, just as when benchmarking the performance of an annealer on individual problem instances.
The performance of the overall hybrid algorithm is thus itself probabilistic and assessed in a similar fashion.

Finally, we reiterate that such a hybrid approach can, of course, only provide a quantum speedup if a raw quantum speedup exists.
The existence of such speedups for practical problems remains heavily debated, but the purpose of the hybrid approach is to exploit such an advantage when or if it is present.

\section{Case study: Dynamically weighted maximum-weight independent set}
\label{sec:caseStudy}

To illustrate the proposed hybrid approach, we discuss in detail a concrete example both from a theoretical and experimental viewpoint.
We first present the problem, which is intended as a proof-of-concept example rather than one of any particular practical application, before discussing an experimental implementation on a D-Wave quantum annealer and analysing the results of this experiment.

Our problem is based on a variant of the well-known independent set problem, the maximum-weight independent set (MWIS) problem. 
More precisely, we consider the question of solving many instances of this problem with different (dynamically assigned) weights on the same graph.

\subsection{Maximum-weight independent set} \label{sec:MWIS}

Recall that an \emph{independent set} $V'$ of vertices of a graph $G=(V,E)$ is a set $V' \subseteq V$ such
that for all $\{u,v\} \in E$ we have $\{u,v\} \not\subseteq V'$.  

\begin{samepage}
\noindent{\bf Maximum-Weight Independent Set (MWIS) Problem}:\\[1.5ex]
\begin{tabular}{ll}
\emph{Input:} & A graph $G=(V,E)$ with positive vertex weights $w:V \rightarrow \mathbb{R}^+$. \\
\emph{Task:} & Find  an independent set $V' \subseteq V$ 
such that maximises
$\sum_{v \in V'} w(v)$ \\
& over all independent sets of $G$.
\end{tabular} 
\end{samepage}

Note that the number of vertices in a maximum weighted independent set may be of smaller size then the number 
for its maximum independent set.
For example, consider the weighted graph shown in Figure~\ref{fig:MWISexample}(a).  
The vertices $\{v_2,v_4\}$ have total weight $9$, while the larger set $\{v_0,v_1,v_3\}$ has only total weight $8$.

\begin{figure}[ht]
\begin{center}
\begin{tabular}{c@{\hspace*{1cm}}c}
\includegraphics[width=2.8in]{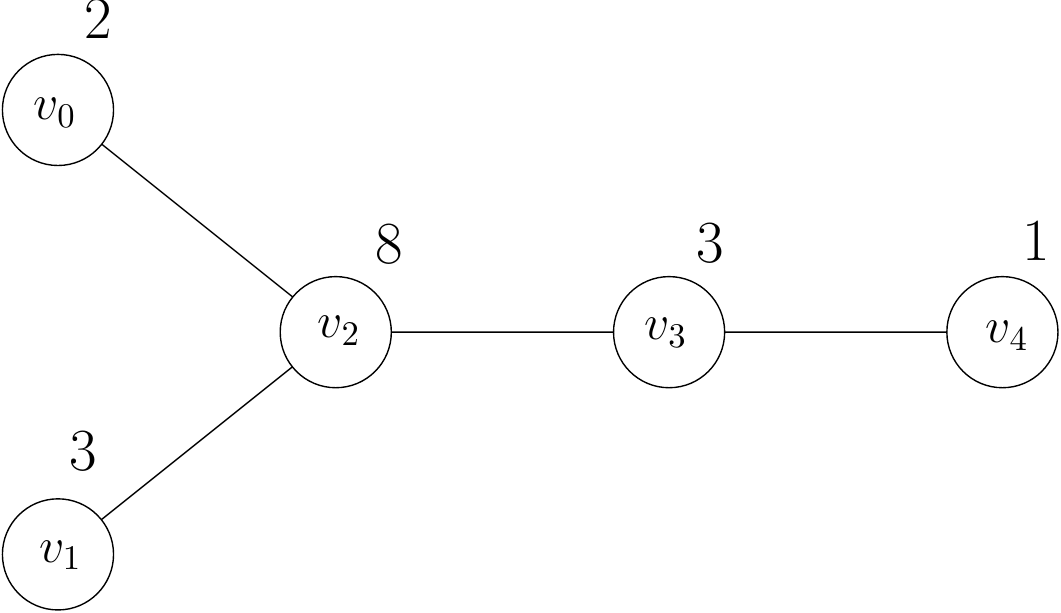}
& 
\large
\def\arraystretch{1.1}
\setlength{\tabcolsep}{3pt}
\begin{tabular}[b]{l|ccccc}
\multicolumn{1}{c}{} & $v_0$ & $v_1$ & $v_2$ & $v_3$ & $v_4$ \\ \cline{2-6}
$v_0$ & -2 &  0& 12&  0&  0 \\
$v_1$ & 0  & -3& 12&  0&  0 \\
$v_2$ & 0  &  0& -8& 12&  0 \\
$v_3$ & 0  &  0&  0& -3& 12 \\
$v_4$ & 0  &  0&  0&  0& -1 \\
\end{tabular}
~\\
&\\
(a) & (b)
\end{tabular}
\end{center}
\caption{An example of (a) a vertex-weighted graph and (b) its MWIS QUBO matrix (cf.\ Section~\ref{sec:quantumSoln}).}\label{fig:MWISexample}
\end{figure}

The general MWIS problem is NP-hard since it encompasses, by restriction, the well-studied non-weighted version~\cite{GJ79}. 
One should note, however, that for graphs of bounded tree-width, the MWIS problem is polynomial-time solvable using standard 
dynamic programming techniques (see Ref.~\cite{Marx10}).

We finish the presentation of the MWIS problem by mentioning an important application of it that was studied in Refs.~\cite{IS:app2} and~\cite{IS:app}. 
Hence, although the example we presented is intended simply as a proof-of-concept, it is not far removed from computational problems of interest. 
Suppose we have a wireless network consisting of several nodes and each node has a certain amount of data it needs to transfer. 
The problem consists in finding the set of nodes that should be given permission to transfer so that the total amount of data output is maximised under the condition that none of the transmissions can interfere with each other. 
If the vertices of the graph $G = (V,E)$ are devices in the network, the weight associated with each node represents the amount of data it needs to transfer and each edge in $E$ codes the potential interference between its two endpoints (so that only one of them can be transferring at a given time), then finding the optimal schedule for transmission is equivalent to finding the maximum-weight independent set of $G$.

\subsection{Dynamically weighted MWIS}

Although the MWIS can be readily transformed into a QUBO problem (as we show below), by itself it is not directly suitable for the hybrid approach we proposed. 
However, a simple variation that we propose here is indeed suitable.

Consider the network scheduling problem presented in the previous subsection. 
Suppose that each node in the network now has multiple messages it needs to send with various sizes, 
but the underlying structure of the graph remains the same (i.e., the same set of devices with unchanged potential interference), but the weight associated with each node will now change over time. 
Finding the optimal transmission schedule over time in this network is the same as finding the maximum weighted 
independent set of the graph with multiple weight functions. 

Formally, we have the following problem:

\begin{samepage}
\noindent{\bf Dynamically Weighted Maximum-Weight Independent Set (DWMWIS) Problem}:\\[1.5ex]
\begin{tabular}{ll}
\emph{Input:} & A graph $G=(V,E)$ with a set of weight functions $\mathcal{W}=\{w_1, w_2,$\\ & $\ldots, w_m\}$ where $w_i:V \rightarrow
\mathbb{R}^+$ for $1 \leq i \leq m$. \\
\emph{Task:} & Find independent sets $V_i \subseteq V$ that maximise 
$\sum_{v \in V_i} w_i(v)$ for each\\ &  $1 \leq i \leq m$.
\end{tabular} 
\end{samepage}

This problem is to solve the MWIS problem on $G$ for each of the $m$ weight assignments $w_i\in \mathcal{W}$.

For $m=1$ we obtain again the MWIS problem, but for larger $m$ the problem is suitable for our hybrid approach.

\subsection{Quantum solution}\label{sec:quantumSoln}

We now provide a QUBO formulation for the MWIS Problem.
Fix an input graph $G=(V,E)$ with positive vertex weights $w:V \rightarrow \mathbb{R}^+$.
Let $W=\max \{w(v) \mid v \in V \}$ and let $S>W$ be a ``penalty weight''.
We build a QUBO matrix of dimension $n=|V|$ such that:

\begin{equation}\label{MWISQUBO}
Q_{(i,j)} = \left\{
\begin{tabular}{cl}
0,       & \mbox{if } $i>j \mbox{ or } \{i,j\} \not\in E,$ \\
$-w(v_i)$, & \mbox{if } $i=j,$ \\
$S$,       & \mbox{if } $i<j \mbox{ and } \{i,j\} \in E.$ \\
\end{tabular}    
\right.
\end{equation}

\medskip

\begin{Theorem}
The QUBO formulation given in~\eqref{MWISQUBO} solves the MWIS Problem.
\end{Theorem}
\begin{proof}
Let $\xvec$ be a Boolean vector corresponding to an optimal solution to the QUBO formulation~\eqref{MWISQUBO}.
Let $D(\xvec) = \{v_i \mid x_i=1\}$ be the vertices selected by $\xvec$.

If $D(\xvec)$ is an independent set then $-x^*=-\xvec^T Q \xvec$ is its weighted sum.  
For two different solutions $\xvec_1$ and $\xvec_2$, which correspond to independent sets,
the smallest value of 
$\xvec_1^T Q \xvec_1$ and $\xvec_2^T Q \xvec_2$ is better.

Now assume $D(\xvec)$ is not an independent set. We will show that the objective function corresponding to $D(\xvec)$ can be improved. Indeed, 
since $D(\xvec)$ is not independent there must be two vertices $v_i$ and $v_j$ in $D(\xvec)$ such that $\{v_i,v_j\}$ is an edge in the graph.
Let $\xvec_1 = \xvec$ but set $x_i=0$, i.e.\ $D(\xvec_1)= D(\xvec) \setminus \{i\}$.  
We have $\xvec_1^T Q \xvec_1 < \xvec^T Q \xvec -W+w(v_i) \leq \xvec^T Q \xvec$.
(Note the second inequality is saturated if and only if  $v_i$ is a pendant vertex attached to $v_j$.)
We can repeat this process on improving $\xvec$ to $\xvec_1$ until we get an independent set. Thus the optimal value
of the QUBO holds for some independent set.  By the conclusion of the  second paragraph of this
proof, we know that a maximum weighted independent set corresponds to $x^*$.
\end{proof}

In Figure~\ref{fig:MWISexample}(b) we give the QUBO matrix for the example in Figure~\ref{fig:MWISexample}(a) with penalty entries~\cite{Dahl:2013aa,dwavebroadcast2016} $P=12 > W=8$. 
It is easy to see that with $\xvec = (0,0,1,0,1)$ we have the minimum value $x^* =\xvec^T Q \xvec = -9$. 
The maximum total weight is thus indeed $-x^*=9$, as expected. 

As a sanity check of the practicality of this solution on real quantum annealing machines, we implemented it on a   D-Wave 2X device.  
For this example it is easy to see that the graph in Figure~\ref{fig:MWISexample}(a) is a subgraph of $K_{4,4}$, hence a trivial embedding is possible.\footnote{We took, for example, the embedding $[v_0\to 0, v_1\to 1, v_2\to 4, v_3\to 2, v_4\to 7]$ into the first bipartite block of the Chimera graph shown in Figure~\ref{fig:ChimeraGraph}.}
The algorithm gave the expected optimal answer of $\{v_2,v_4\}$ approximately two-thirds of the time, and the non-optimal answer of $\{v_0, v_1, v_3\}$, a third of the time; occasionally other results, such as $\{v_2\}$ or $\{v_0,v_1,v_4\}$ were obtained, although such occasional incorrect solutions are not unexpected for quantum annealers.
Further details of the implementation, including source code, are available online in Ref.~\cite{Abbott18cdmtcs}.

In order to adapt the MWIS solution above to the DWMWIS problem, note that the locations of the non-zero entries of the QUBO formulation~\eqref{MWISQUBO} depend only on the structure of the graph and not on the weight function $w$.  
Thus, in order to solve the DWMWIS problem, for each weight assignment $w_i$ the same embedding of the graph into the D-Wave physical graph can be used, meaning that a hybrid algorithm based around the MWIS solution above can readily be implemented.

More specifically, following the hybrid algorithm described in Section~\ref{sec:hybridAlgGen} for instances $P_1, \ldots, P_m$ (where each $P_i$ uses weight function $w_i$), we perform the embedding once (entailing a time $t_{\text{embed}}(P_1)$) and then solve the MWIS problem for each weight assignment $w_i$ (taking times $t_{\text{proc}}(P_i)$) using the QUBO solution outlined above. 
Note that the iteration times $t_2(P_i)$, $1 \leq i \leq m$, in Eq.~\eqref{eqn:hybridBreakdown} thus correspond to the time to read in
and alter the coupling weights in the QUBO matrix.

\subsection{Classical baseline}

The main objective of studying the DWMWIS example in detail is to exhibit experimentally the advantage that the hybrid approach can provide over a standard annealing-based approach.
Nonetheless, it is helpful to further compare this to the performance of a classical baseline algorithm for comparison and to help highlight this advantage, even if we do not necessarily expect to see an absolute quantum speedup from the hybrid algorithm.

As we discussed in detail in Section~\ref{sec:comparisonToClassical}, one should ideally compare the hybrid algorithm against the best available classical algorithm for the same problem.
However, since our primary concern is not to show an absolute quantum speedup, and studying more closely the performance of various classical algorithms for the DWMWIS problem is somewhat beyond the scope of the present article, we will use a generic classical algorithm based on a Binary Integer Programming (BIP) formulation of the MWIS problem for illustrative purposes.
Both quantum annealing and BIP can be seen as types of generic optimisation solvers.
By using such a baseline, we also mimic how an engineer would map a new hard problem to a well-tuned optimisation solver (a SAT-solver or IP-solver being two natural generic choices). 
This process mimics the D-Wave model of requiring a polynomial-time reduction to the Ising/QUBO problem, which the quantum hardware solves, and allows us to compare similar approaches, even if for certain problem instances their very genericity may make them suboptimal.

To this end, for a given input graph $G=(V,E)$ with positive vertex weights $w: V \rightarrow \mathbb{R}^{+}$, we construct a BIP instance with $n=|V|$ binary variables as follows.
To each vertex $v_i$ in $G$ we associate the binary variable $x_i$, and for notational simplicity we will denote the collection of variables $x_i$ by a binary vector $\xvec = (x_0, x_1, \cdots, x_{n-1})$.
We thus have the BIP problem instance:
\begin{equation} \label{IP}
\begin{array}{ll@{}l}
\text{maximise }  \displaystyle\sum\limits_{v_i \in V} w(v_i)x_{i} \\
\text{subject to } \displaystyle x_i + x_j \leq 1 \mbox{ for all }  \{v_i,v_j\} \in E.
\end{array}
\end{equation}

Each constraint in~\eqref{IP} enforces the property that no adjacent vertices are chosen in the independent set while the objective function ensures an independent set with maximum sum value is chosen. 
Assuming we have the binary vector $\xvec$ which yields the optimal value of objective function~\eqref{IP}, we take $D(\xvec) = \{v_i \mid x_i=1\}$ to be the set of vertices selected as the maximum weighted independent set.
\\
\begin{Theorem}
The BIP formulation given in (\ref{IP}) solves the MWIS problem.
\end{Theorem}
\begin{proof}
First, we show that $D(\xvec)$ is an independent set if and only if all the constraints in~(\ref{IP}) are satisfied. This is indeed the case as
if all the constraints are satisfied, then for each $\{v_i,v_j\}$ in $E$, at most one of them is in $D(\xvec)$ by its definition. 
On the other hand, if any one of the constraint is not satisfied, then it means $v_i$ and $v_j$ are both chosen,
thus $D(\xvec)$ is not 
an independent set. 

Now, let $\xvec$ be a binary vector corresponding to an optimal solution of BIP formulation~(\ref{IP}).
Let $D(\xvec) = \{v_i \mid x_i=1\}$ be the vertices selected by $\xvec$. Since $\xvec$ is the optimal solution, we already have 
all the constraints of~(\ref{IP}) satisfied and $D(\xvec)$ is therefore a valid independent set. The objective function will 
ensure that the selected independent set has the maximum value sum.
\end{proof}

The classical baseline%
\footnote{Our local linux machine, running Fedora 25 OS, consisted of an Intel Haswell i7 4.0GHz (overclocked to 4.5GHz) with 32GB DDR3
2400MHz RAM.}
we use in the analysis presented in the remainder of this section is based on an implementation of the BIP formulation in Sage
Math~\cite{sagemath}, which has a well developed and optimised Mixed Integer Programming library.
Note that this is an exact solver for BIP problems, whereas an annealer can also be used to find good approximate solutions. However, since we are using the TTS metric we thus treat the quantum annealer as an exact solver too, thereby ensuring a fair comparison. (If a TTT metric---see Footnote \ref{fn:TTT}---were instead considered one would need, for fairness, to compare the annealer to a classical approximation algorithm.)

To ensure that a fair comparison with the hybrid algorithm is possible, we formulate the classical algorithm for the overall DWMWIS problem such that \emph{the set of constraints in the BIP formulation is only computed once} (cf.~the discussion in Section~\ref{sec:hybridAlgGen}).
This is possible since (in analogy with the need to only perform the embedding once in the quantum solution) the changing weights do not change the constraints of the BIP formulation, and we make use of this to reuse parts of the computation where possible.
Note that the Sage environment contains a simple Python front-end interface to one
of many (Mixed) IP-solvers which are often written, optimised and compiled from
C.  We used the default Gnu GLPK as the back-end library but many popular
commercial solvers like COIN-OR, CPLEX or GUROBI could be equally used.  For
our small input instances, the classical solver choice would not matter much;
the scaling behaviour would be the same for our chosen illustrative NP-hard
problem.

\subsection{Experimental framework}

To study experimentally the performance of the hybrid DWMIWS algorithm, we compare the performance of three algorithms on a selection DWMWIS problem instances: the ``standard'' quantum algorithm, in which the embedding is re-performed for each weight assignment; the hybrid DWMWIS algorithm; and the classical BIP-based algorithm described above.

To this end we analyse the algorithms on a range of different graphs, initially choosing $156$ graphs from a variety of common graph families with between 2 and 126 vertices (which initial testing suggested should place most of them within the capabilities of the quantum annealer we used).
These graphs, including the so-called common graphs in SageMath~\cite{sagemath} with no more than 126 vertices and representatives from several well-known families of graphs, are natural and well studied examples spanning a range of sizes and with varying properties with which to test the performance of our hybrid DWMIS algorithm. Moreover, they were also used in Refs.~\cite{Dinneen:2017} and~\cite{dinneen19} to study other quantum annealing algorithms, allowing our results to be comparable to those.
The full list of graphs and some of their basic properties (order, size) can be found in the summary of results in 
Appendix~\ref{app:resultSummary}.
Each graph was used to generate a single DWMWIS problem instance with $m=100$ weight assignments, each randomly generated as floating point numbers rounded to 2 decimal places within the range $[0.0,1.0)$ using the default pseudo-random generator in Python.\footnote{This choice of weight distribution was made for simplicity, but one would expect similar behaviour for other distributions. In practice, using the full range of possible weights leads to better quantum annealing performance, so other distributions might require rescaling to optimise performance, adding additional technical---but not fundamental---complications.}
Although the choice of $m$ of weight assignments is somewhat arbitrary, our choice was made by the need to balance the ability to solve sufficiently large problems to be able to negate the embedding time against the limited access we had to the quantum annealer.
The problem instances were generated as standard adjacency list representations using SageMath~\cite{sagemath} with random weights assigned.

The hybrid DWMWIS algorithm outlined in Section~\ref{sec:quantumSoln} was implemented on a D-Wave 2X quantum annealer with $1098$ active physical qubits~\cite{DWave2X}.
Note that this is significantly more qubits than are needed to embed any of the graphs we consider. However, as we will see, for the larger graphs we considered D-Wave already struggled to find optimal solutions making further analysis impossible; indeed, this is why we did not initially select larger graphs for analysis despite the ability to embed them into the hardware graph.
The same procedure is used for the ``standard'' quantum algorithm, except the cost of the embedding is incurred for each weight assignment (as per Section~\ref{sec:hybridAlgGen}).
Full details of the implementations, data and results (i.e., source code, problem instances and outputs) are available online in Ref.~\cite{Abbott18cdmtcs}.

Since we are primarily interested in negating the impact of the embedding process in general applications, we made use of D-Wave's heuristic embedding algorithm~\cite{DWave:2013aa} to embed each logical graph in the physical graph.
While specialised embedding algorithms may be more effective in certain scenarios, the overall hybrid approach would still be applicable, and by adopting a generic algorithm our results have wider relevance.
Each graph was embedded 10 times to estimate $t_\text{embed}$ for each problem instance.
Unfortunately, due to the large number of samples often required to be run for each problem and restrictions on access to the annealer, we were unable to perform a full analysis with each embedding (recall the embedding is non-deterministic) and instead performed the analysis for a single such embedding.
This introduces a potential systematic error since the embedding generally affects the solution quality to some degree;  we will discuss this further in the analysis that follows.
The orders of the embeddings we used (i.e., the number of physical qubits required), which are useful in understanding the performance of the quantum annealer on individual problems, as well as the maximum chain length required in embedding are also given in Appendix~\ref{app:resultSummary}.
With the embeddings obtained, the couplings between the physical qubits were determined using the default approach of evenly distributing the logical couplings along chains provided by the D-Wave solver API.

Operational parameters for the D-Wave 2X device were determined via an initial testing round (see Refs.~\cite{DWave:timing} and~\cite{DWave:python} for further information on D-Wave timing parameters).
In line with previous research~\cite{Hen:2015aa,King:2014aa,Ronnow:2014aa,Venturelli:2015aa} (cf.~Section~\ref{sec:measure_time}) we found the minimal annealing time of $20\mu\text{s}$ to be optimal for all the graphs considered.
The programming thermalisation time, which specifies how long the quantum processor is allowed to relax thermally after being programmed with a QUBO problem instance, was chosen as its default value of 1000$\mu$s, as this was seen to produce satisfactory results.
Between anneals, the processor must similarly be allowed to thermalise, and the default $50\mu$s delay was used.
Reading out the result of each anneal takes 309$\mu$s on the   D-Wave 2X device, so this readout time (and to a lesser extent the thermalisation) dominated the actual annealing time.
With minor additional low level processing taken into account, each annealing ``sample'' has a fixed time of $380.2\mu$s.
Although the actual annealing time of 20$\mu$s was a minor part of each annealing cycle, this is likely to change in the future as larger problems necessitating longer annealing times become accessible. 
Moreover, future generations of the machine could have shorter relaxation periods and faster readout times (at least relative to the annealing time, if not in absolute terms) as the physical engineering of the processor is better developed~\cite{King:2014aa,King17}.

Finally, our tests were run with D-Wave's post-processing optimisation enabled.
While this adds a small overhead in time, this is well within the spirit of hybrid quantum-classical computing, and allowed us to solve more problems.
This post-processing method processes small batches of samples while the next batch is being processed~\cite{DWave:postprocessing}.
This ensures that it only contributes a constant overhead in time for each MWIS problem instance \emph{independent of the number of samples (and thus of $k_{99}$)}.

To estimate the TTS times $T_H$ and $T_\text{std}$ described in Section~\ref{sec:hybridAlgGen}, one must first estimate $k_{99}$, as defined in Eq.~\eqref{eq:k99}, for each weight assignment $w_i$.
This is done by estimating the probability of success $s_i$ for each such case as $N_\text{opt}/N_\text{total}$, where $N_\text{total}$ is the number of annealing cycles performed, while $N_{\text{opt}}$ denotes the number of times an optimal solution was found.
To determine this ratio accurately for each weight assignment, each problem instance was initially run twice with 1000 samples.
Problem instances for which an optimal solution was not found several times for every weight assignment were run a further 5 times; the hardest instances were eventually run a further two times with 2000 samples per run and, for one difficult graph (the complete bipartite graph $K_{12,12}$) a further 14 runs of 2000 samples.
By performing many runs (and since each weight assignment is considered separately), random noise due primarily to analogue programming accuracy is largely reduced, and $k_{99}$ is estimated more accurately.

Some problem instances remained unsolved after these runs (i.e., there was at least one weight assignment $w_i$ for which an optimal solution was never found so that $k_{99}$ was undefined) and such problem instances had to be abandoned; indeed, this was the limiting factor in the size of graphs analysed, preventing us from considering larger problems.
As a result, the initial 156 graphs were reduced to 124 for which a running time could be computed and analysed.
These graphs that we originally selected but for which further analysis could not be performed are listed separately for reference in Appendix~\ref{app:resultSummary} (see Table~\ref{table:resultSummary2}).
The fact that such cases were not uncommon despite the relatively modest size of the graphs (even the largest embedded graph required only 280 of the 1098 available physical qubits) highlights limitations of the current state of quantum annealing on more traditional (and, potentially, practical) computational problems.

\subsection{Results and analysis}

For each DWMWIS problem instance (i.e., for each graph $G$) the times $T_H$ and $T_\text{std}$ were calculated, following the approach described in Section~\ref{sec:hybridAlgGen}, as 

\begin{equation*}
	T_H=t_\text{embed} + \sum_i \big(t_\text{prog}(P_i) + k_{99}(P_i)t_\text{anneal} + t_\text{post}(P_i)\big)
\end{equation*}
and 
\begin{equation*}
T_\text{std} = \sum_i \big(t_\text{embed} + t_\text{prog}(P_i) + k_{99}(P_i)t_\text{anneal} + t_\text{post}(P_i)\big),	
\end{equation*}
where $k_{99}(P_i)$ is the $k_{99}$ value for weight assignment $w_i$ and $t_\text{anneal}=309\mu$s.
As noted in Section~\ref{sec:hybridAlgGen}, $T_\text{std}$ may be reduced by a small constant factor by exploiting classical parallelism, so $T_\text{std}$ as defined here constitutes an upper bound on the time of a traditional quantum annealing approach.
Both $t_\text{prog}(P_i)$ and $t_\text{post}(P_i)$ are of the order of $20$ms (although the latter varies by an order of magnitude more than in the former over different problem instances and runs).
Note that the processing time $t_\text{proc}$ defined earlier is, for this approach to the DWMWIS problem, given by $$t_\text{proc}=t_\text{prog}(P_i) + k_{99}(P_i)t_\text{anneal} + t_\text{post}(P_i).$$
The classical time $T_C$ was taken as the processor time for the classical algorithm described earlier.

A detailed summary of the overall times for each graph is given in Appendix~\ref{app:resultSummary}.
These results are summarised in Figures \ref{fig:TH_v_Tstd_v_TC}(a) and \ref{fig:TH_v_Tstd_v_TC}(b), which show how the hybrid times $T_H$ compare to both $T_\text{std}$ and $T_C$.
Error bars are calculated from the observed variation in $t_\text{embed}$, the number of optimal solutions found $N_\text{opt}$, and the post-processing time $t_\text{post}$.
Of these, the error in $t_\text{post}$ is the dominant factor, and largely arises from the uncontrollability of the post-processing environment, which is performed remotely within the D-Wave processing pipeline.
However, this variation did not result in any significant variation in success probability of the annealing, so it seems the computational effort expended on post-processing was nonetheless constant.
Indeed, we note that in some earlier runs the post-processing was performed 20 times faster with no noticeable change in the quality of solution.
Given that post-processing contributes non-negligibly to $T_H$ and $T_\text{std}$, this could significantly effect the overall times. We discarded these results to present a conservative analysis and the overall conclusions are not affected by this, but we note that, with increased control of the classical post-processing, the quantum times could be significantly reduced.

As noted in the previous section, practical and logistical constraints prevented us from taking the variation due to different embeddings of each graph fully into account.
To assess the possible magnitude of this effect, we tested one relatively difficult graph (Shrikhande) and found that consideration of the embedding roughly tripled the error in $T_H$, changing the value from $12,800^{+370}_{-240}\mu$s 
to $15,300\pm 1,280\mu$s, with the average size of the embedding being $67$ physical qubits but with a standard deviation of $6.5$, explaining much of this variation.  
While this variation would thus generally be a significant source of error, the variation it induces will not be large enough to affect any of our conclusions significantly, even if the inability to take this into account is admittedly regrettable.

\begin{figure}[t]
\begin{center}
\begin{tabular}{cc}
\includegraphics[width=0.5\textwidth]{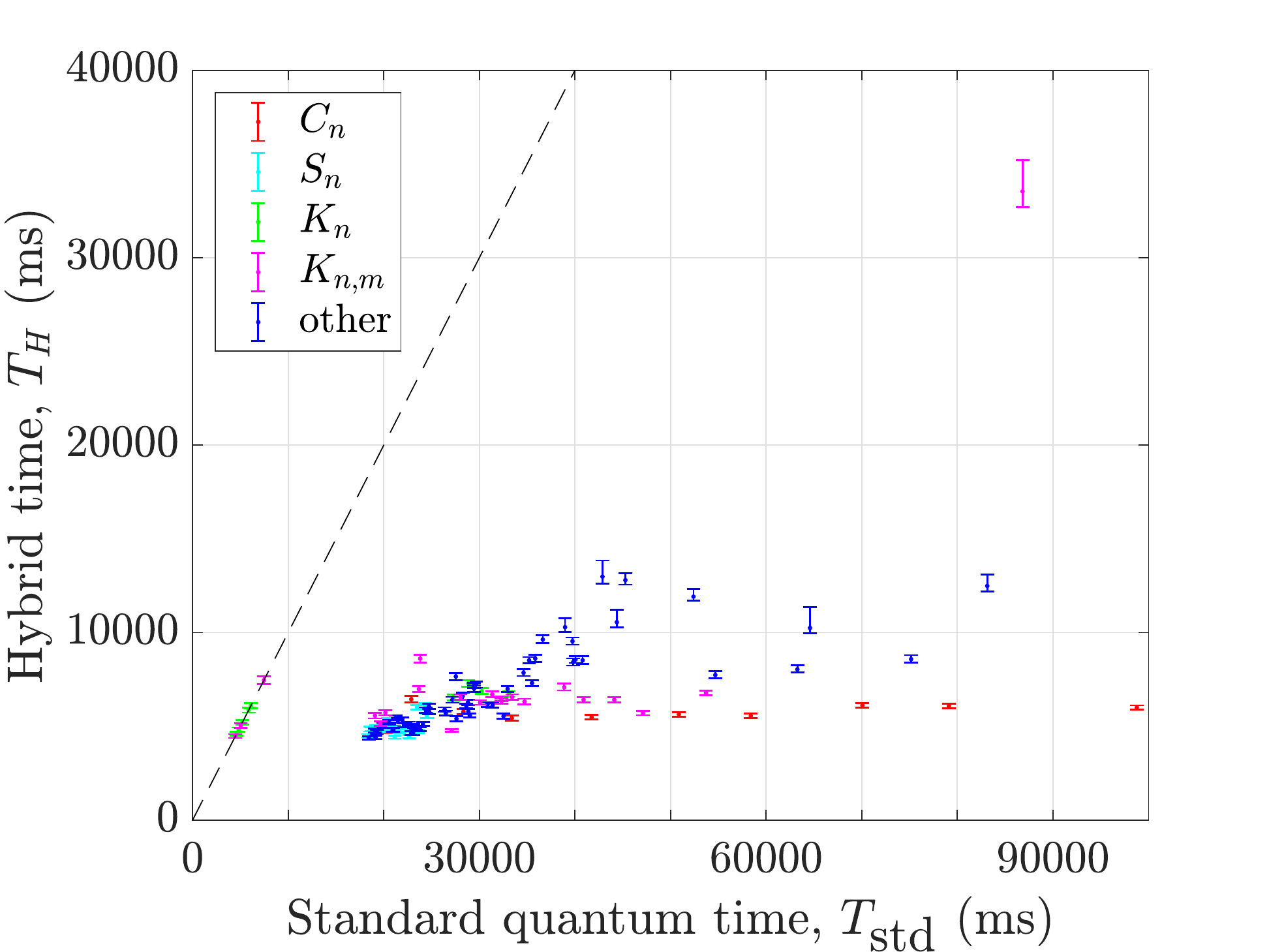} & \hspace{-5mm} \includegraphics[width=0.5\textwidth]{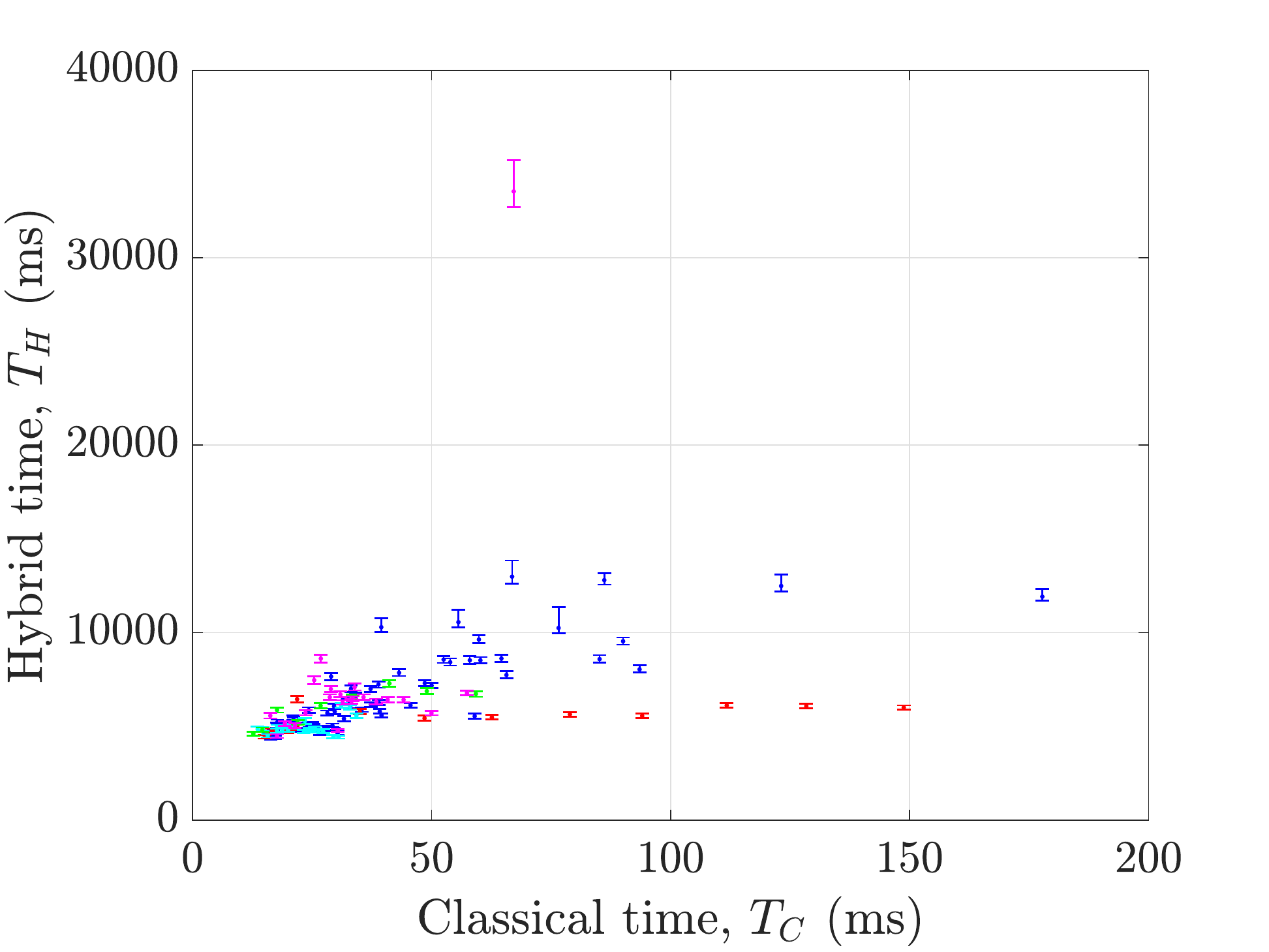}\\
(a) & (b)
\end{tabular}
\caption{(a) An upper bound for $T_\text{std}$ against $T_H$. The dashed line is $T_H = T_\text{std}$; the points falling on the line are correspond to graphs for which a trivial embedding was possible. (b) $T_C$ against $T_H$ for each DWMWIS problem instance. The different colours indicate specific graph families for reference: cycle graphs $C_n$, star graphs $S_n$, complete graphs $K_n$, and complete bipartite graphs $K_{n,m}$. All times are in ms.}\label{fig:TH_v_Tstd_v_TC}
\end{center}
\end{figure}

First and foremost, from the results shown in Figure~\ref{fig:TH_v_Tstd_v_TC}(a) the extent of the advantage of the hybrid approach is  evident.
Indeed, this is to be expected given that, for a given DWMWIS problem, they differ (by definition) by $99\times t_\text{embed}$.
Although this might seem a trivial confirmation of this fact, the results help illustrate the extent of the advantage that the hybrid approach can have for such problems, a consequence of the absolute cost of the embedding.

It is interesting to ask, moreover, how the advantage of the hybrid approach scales.
To understand this, we first look at how $t_\text{embed}$ scales since, as long as this remains a significant compared to the annealing time, this will largely determine the scaling of the hybrid advantage.
Recall, following the discussion of Section~\ref{sec:Chimera}, that although a poor embedding can be quickly found (given enough physical qubits), one may expect, in general, the time required to find a good embedding---as the heuristic embedder we use indeed tries to do---to scale exponentially with the graph order.

This is confirmed in Figure~\ref{fig:embedding}(a), showing $t_\text{embed}$ as a function of the number of vertices in a graph.
Moreover, from the figure one sees that, even for these relatively small graphs, $t_\text{embed}$ quickly approaches 1s.
The figure shows that there is a large variation in the embedding times, due in part to different behaviour on graphs from different families.
Indeed, some graph families are quite naturally easier to embed than others for the heuristic embedder; of course, for a known family of graphs one could generally find a good efficient embedding, but our interest is in understanding the scaling for a generic approach that should work on arbitrary graphs.

The scaling behaviour is clearer to see in Figure~\ref{fig:embedding}(b), which shows $t_\text{embed}$ instead as a function of the embedded graph order, i.e., the number of physical qubits required, since this accounts for much of the difference in difficulty in embedding different graphs.
There, a nonlinear regression analysis shows a much better fit with exponential scaling and, given that the embedded graph order scales at most quadratically with the logical graph order, this confirms the exponential scaling of the embedding time.

\begin{figure}[t]
\begin{center}
	\begin{tabular}{cc}
	\includegraphics[width=0.5\textwidth]{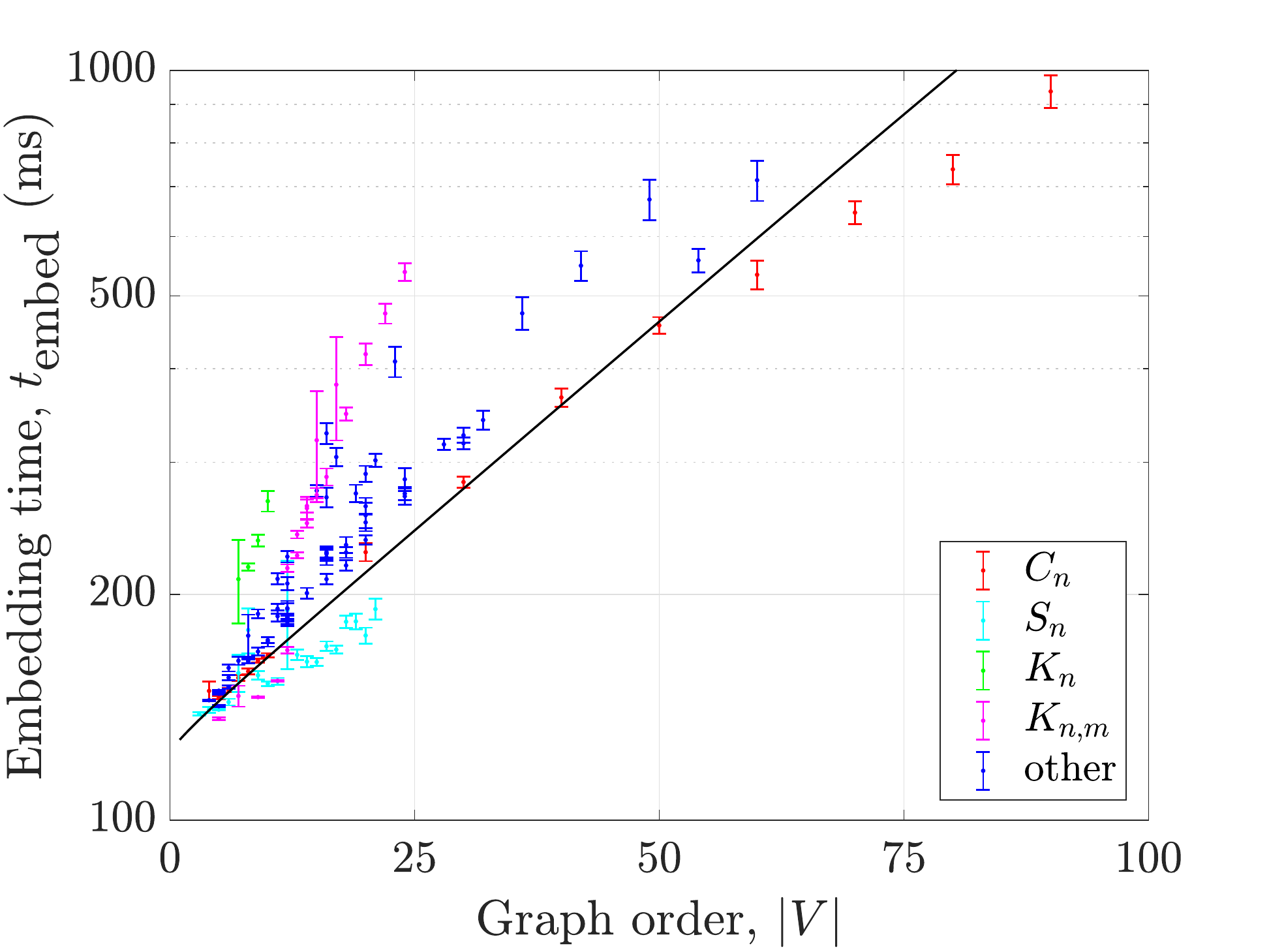} & \hspace{-5mm} \includegraphics[width=0.5\textwidth]{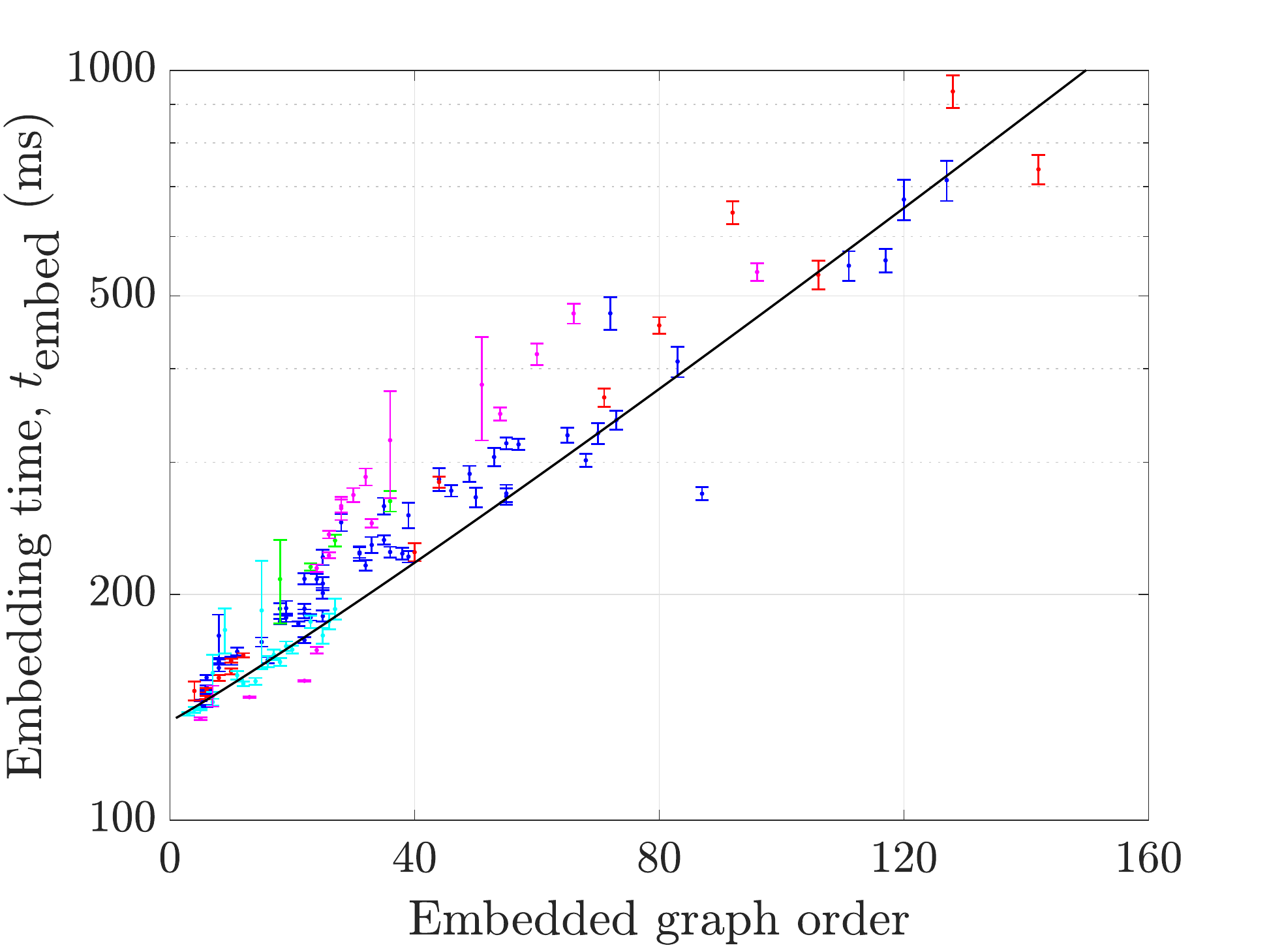}\\
	(a) & (b)
	\end{tabular}

\caption{(a) Plot of graph order $|V|$ against the embedding time $t_\text{embed}$ on a logarithmic time scale. (b) Plot of the order of the embedded graph against the embedding time $t_\text{embed}$ showing a better exponential fit. The colours show particular graph families for reference.}\label{fig:embedding}
\end{center}
\end{figure}

The overall advantage of the hybrid approach will depend not just on $t_\text{embed}$, but how this relates to the rest of the annealing times.
To study this more directly, it is useful to look at the ``hybrid speedup ratio'' $R_H = T_\text{std}/T_H$; the larger this value, the more advantage the hybrid approach provides.
In Figure~\ref{fig:hybridSpeedup}(a) we plot this as a function of the graph order $|V|$.
This shows a general tread of increasing $R_H$ (excepting a handful of points with $R_H=1$ for which a trivial embedding was possible, giving $t_\text{embed}=0$).
Although the complexity of the quantum annealing algorithm is \textit{a priori} unknown, it is expected (as for $t_\text{embed}$) to exhibit exponential scaling with $T_H \propto \exp(k_H\cdot n^{\ell_H})$, and one has $R_H = 1+99 t_\text{embed}/T_H$.
Performing a nonlinear regression with such a model shows that indeed $R_H$ appears to be increasing exponentially.
However, there is significant variation between families of graphs: while $R_H$ grows quickly for the complete bipartite graphs $K_{n,m}$, it is relatively constant for the Star graphs $S_n$. 
Moreover, the trend is dominated by the Cycle graphs $C_n$, which include several of the largest graphs in our problem set, meaning that the fit shown has limited generality.

As for the embedding time, it is thus useful to instead look at how $R_H$ depends on the embedded graph order, and we plot this in Figure~\ref{fig:hybridSpeedup}(b).
Although much variation remains between different families of graphs, the general trend remains and is consistent with an exponentially increasing hybrid speedup, despite the trend no longer being dominated by simple families such as the $C_n$ graphs.
We emphasise, however, that the benefit of the hybrid approach in general will depend on the problem one is solving (e.g., the number of times an embedding would need to be performed in the standard approach) and the way in which $R_H$ scales is likely to change further as newer quantum annealers become available and larger problems become solvable.

\begin{figure}[t]
\begin{center}
	\begin{tabular}{cc}
	\includegraphics[width=0.5\textwidth]{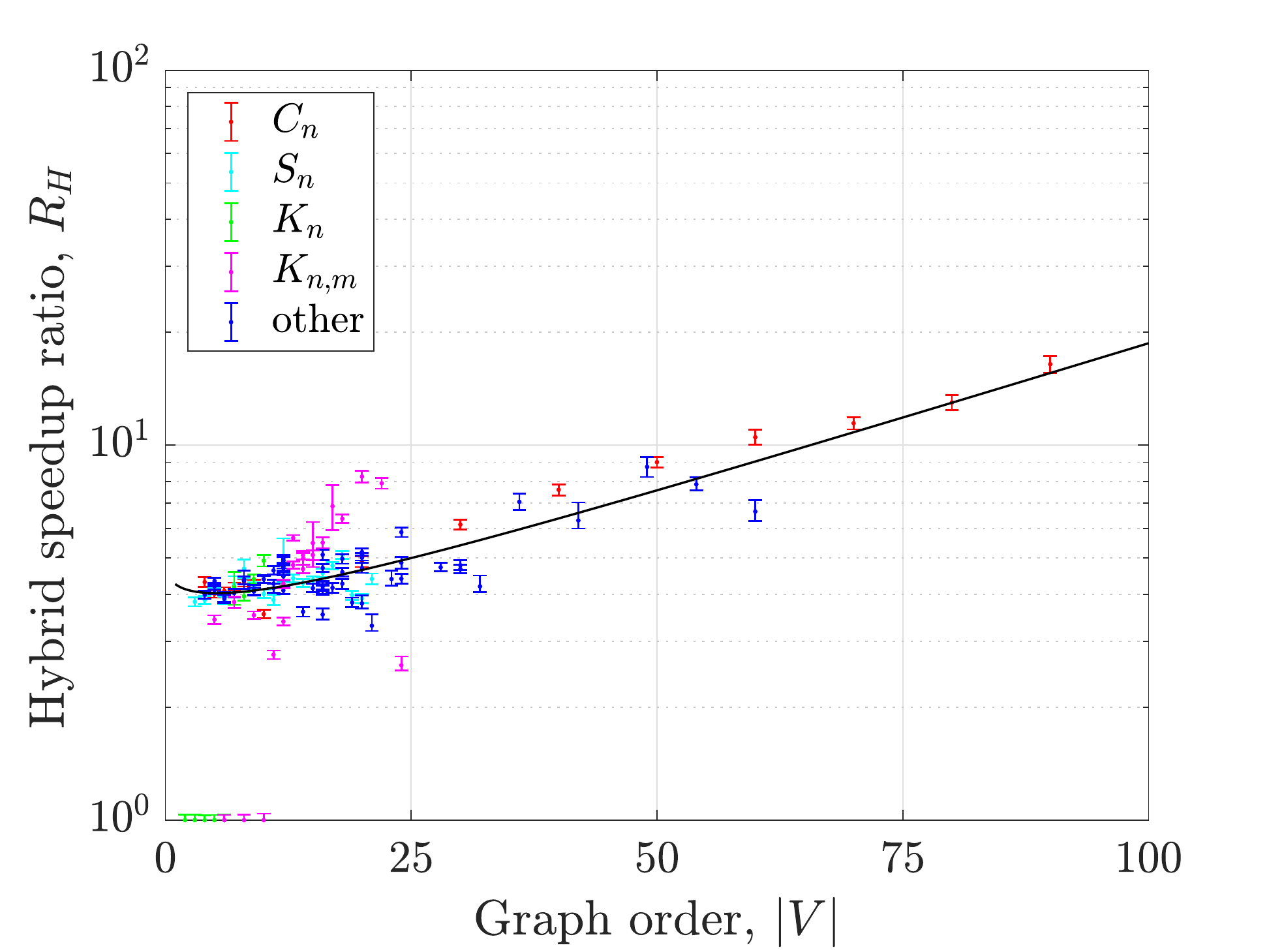} & \hspace{-5mm} \includegraphics[width=0.5\textwidth]{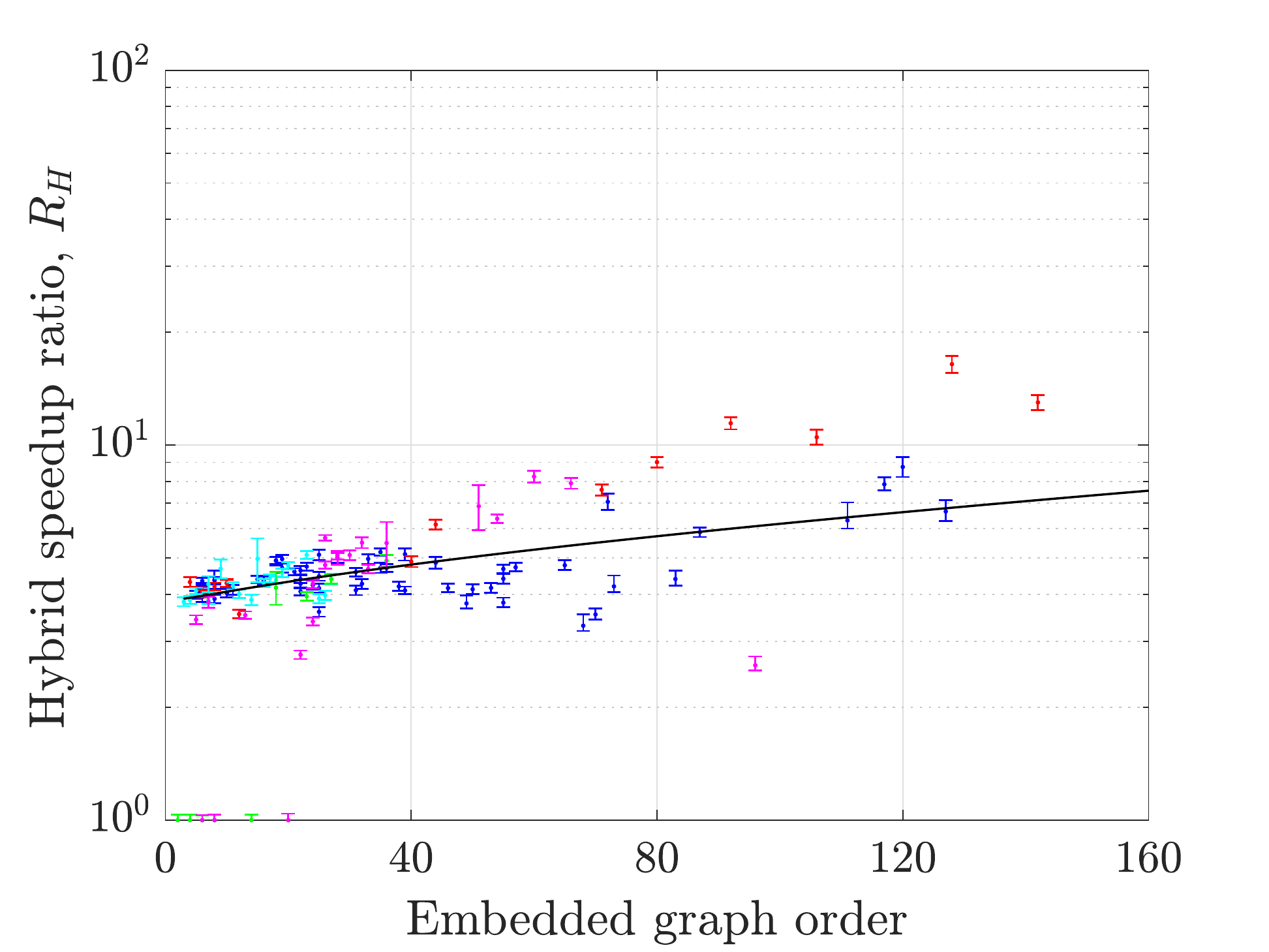}\\
	(a) & (b)
	\end{tabular}

\caption{Logarithmic plots of (a)  graph order $|V|$ against the hybrid speedup ratio $R_H$; and (b) embedded graph order against $R_H$. The colours indicate particular graph families for reference.}\label{fig:hybridSpeedup}
\end{center}
\end{figure}

From Figure~\ref{fig:TH_v_Tstd_v_TC}(b) it is evident that no absolute quantum speedup was observed using the hybrid algorithm, and indeed there is a vast difference in scale between $T_C$ and $T_H$: the ``hardest'' problem was solved classically in less than 200ms, whereas the hybrid algorithm required almost 60 times as much time to solve it correctly.
The inability to observe any raw speedup is hardly surprising when one notes that, even if $k_{99}=1$ and $t_\text{embed}=t_\text{post}=0$, the fact that $t_\text{prog}\approx 20$ms means that that one would have $T_H > 2000$ms.
The programming time thus adds an essentially constant overhead, which would have less of an impact as larger problems (for which $k_{99}$ is much larger) become solvable.

Although no overall raw speedup was observed, the experiment nonetheless illustrated the advantage of the hybrid approach over the standard quantum one which, we recall, was the primary goal.
It is nonetheless interesting to examine the scaling behaviour of the hybrid algorithm in comparison to the classical one, to see
whether there is any indication that a speedup might potentially be obtainable once the overheads (such as the embedding and programming times) are sufficiently negated.
To analyse this more carefully, it will be useful to look at the ``classical speedup ratio'' $R_C=T_H/T_C$, which provides a clearer measure of any potential speedup: a value of $R_C < 1$ thus indicates an absolute speedup for the hybrid algorithm.\footnote{We could equally look at the hybrid speedup $T_C/T_H=1/R_C$, but we choose $R_C$ because it is slightly easier to interpret visually.}

In Figure~\ref{fig:scalingOverall} we show the scaling behaviour of $R_C$ against the graph order $|V|$, which is proportional to the problem size, and the actual number of physical qubits used, i.e., the size of the embedded graph.
While the scaling of an algorithm should generally be studied with respect to problem size, as for the analysis of the hybrid speedup above, Figure~\ref{fig:scalingOverall}(a) shows a large variation between different graphs and, in particular, between the families of graphs within or set of problems.
Since the quantum annealer operates with physical, rather than logical, qubits, one expects that its scaling is better described as a function of the embedded graph size.
By examining $R_C$ as a function of this in Figure~\ref{fig:scalingOverall}(b) one thus removes part of the variation between graphs---although much still remains---and allows this scaling, and thereby the possibility of a potential speedup, to be better analysed.

\begin{figure}[t]
\begin{center}
\begin{tabular}{cc}
\includegraphics[width=0.5\textwidth]{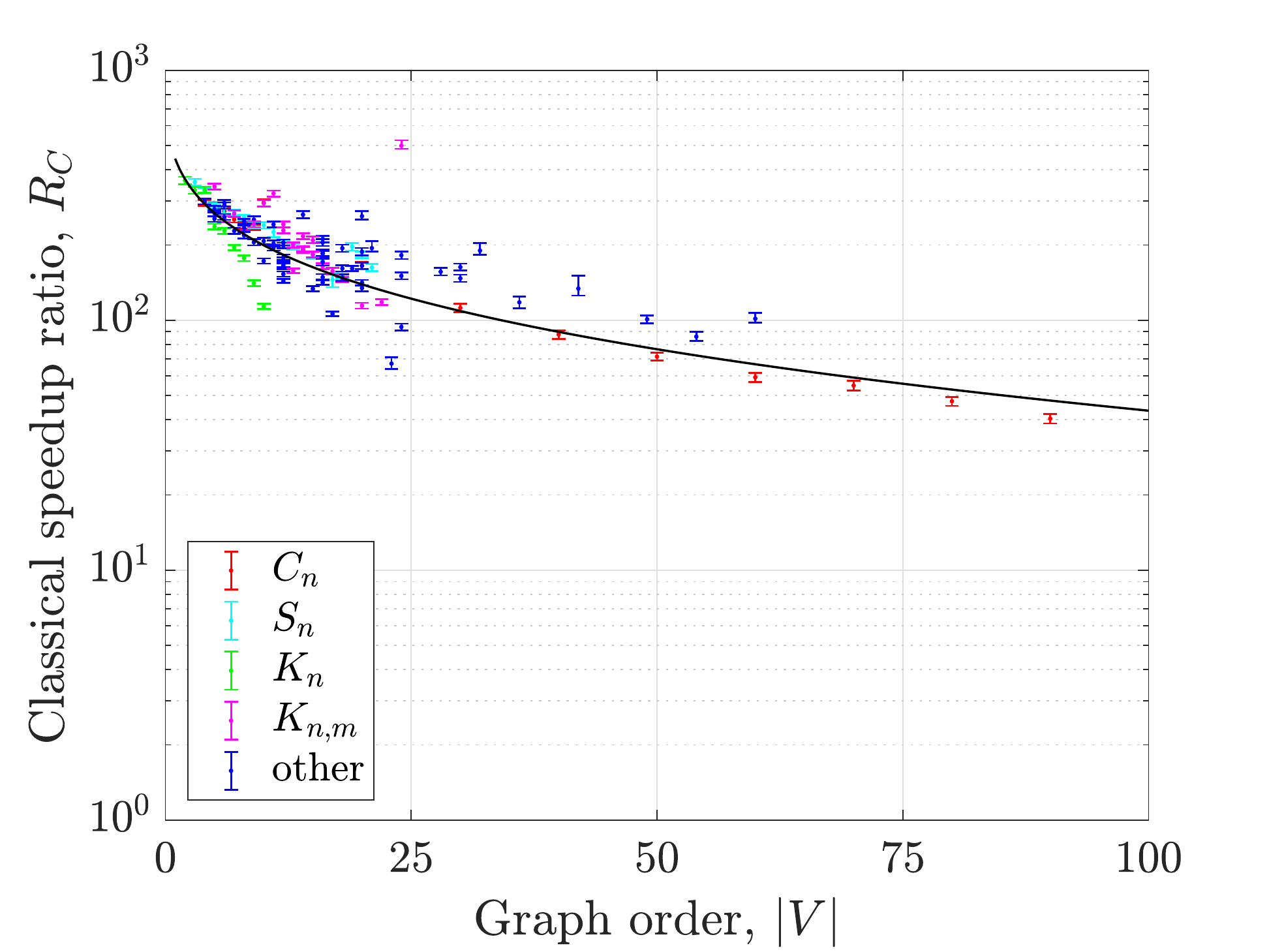} &\hspace{-5mm} 
\includegraphics[width=0.5\textwidth]{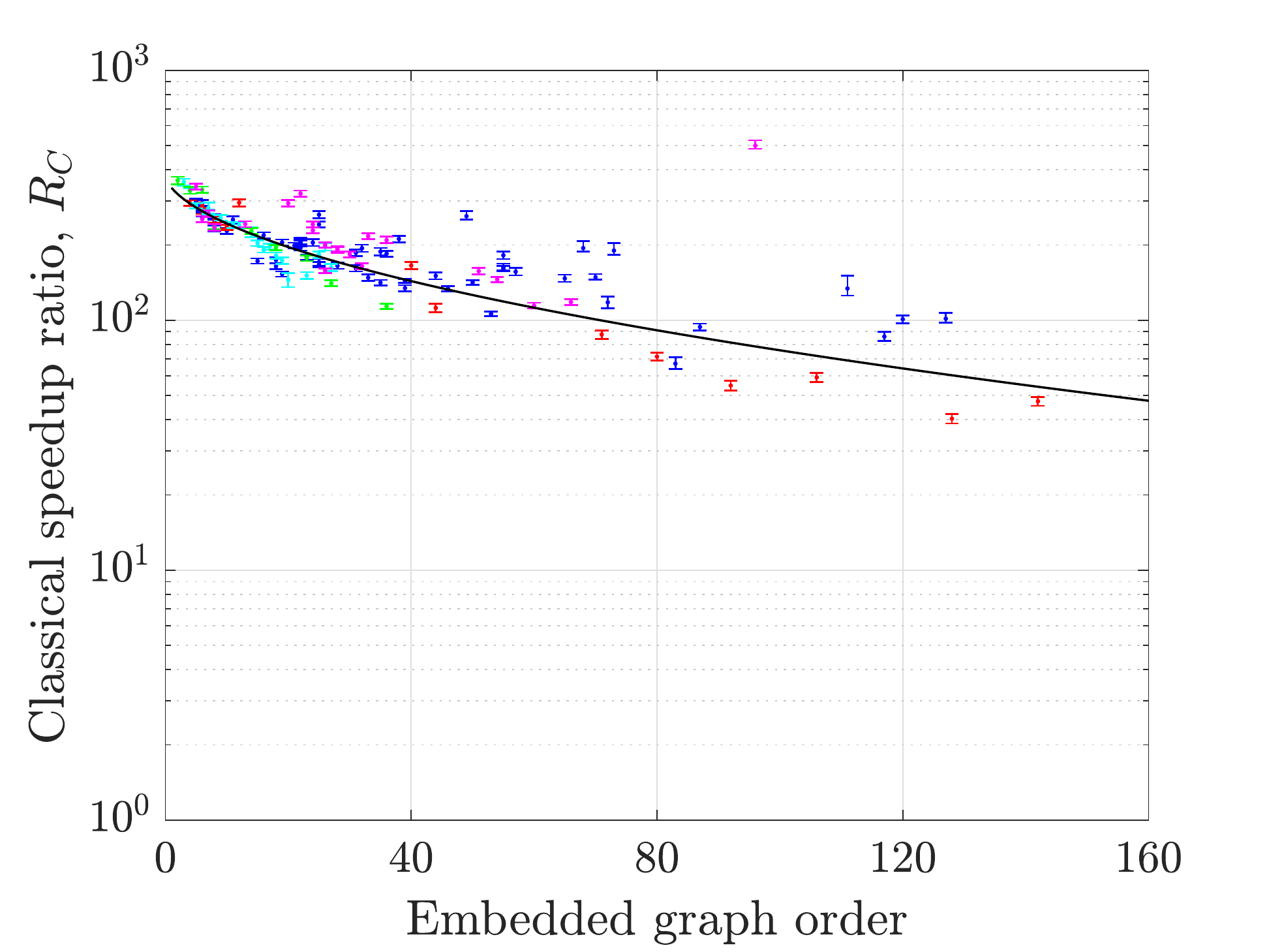}\\
(a) & (b)
\end{tabular}
\caption{Logarithmic plots of the scaling behaviour of the classical speedup ratio $R_C$ for the DWMWIS problem instances: (a) graph order $|V|$ against $R_C$; and (b) embedded graph order against $R_C$. The colours highlight particular graph families as in the previous plots.}
\label{fig:scalingOverall}
\end{center}
\end{figure}

These figures highlight once more the discrepancy between the hybrid and classical times, with the minimum classical speedup observed being $R_C= 40\pm 2$.
Both figures, however, show that $R_C$ decreases with problem size and difficulty, indicating that, for the problem instances tested, the hybrid algorithm exhibited better scaling behaviour than the BIP-based classical algorithm.
Both quantum annealing algorithms and the classical baseline we use (due to it being a relatively generic BIP algorithm) are expected to exhibit some form of exponential scaling (with respect to both the order and embedded graph size, since these differ by at most a quadratic factor), even if the precise complexity of the algorithms is \emph{a priori} unknown.

A nonlinear regression analysis shows that the scaling behaviour of $R_C$ is indeed, with respect to both $|V|$ and the embedded graph order, most consistent with $R_C \propto \exp(k_H\cdot n^{\ell_H})/\exp(k_C\cdot n^{\ell_C})$, for constants $k_H,\ell_H,k_C,\ell_C$, with the hybrid algorithm scaling slower.
With respect to $|V|$ the large variation in performance over different graph families and the fact that the scaling is largely dominated by the larger $C_n$ graphs means that little can be read into the precise form of the scaling.
While much variation remains when viewed as a function of the embedded graph order, the fit is nevertheless better in that case.

It is possible to extrapolate these fits to obtain a very crude estimate of when one might obtain $R_C=1$, at which point the hybrid and classical algorithms require the same amount of time.
One finds that this point is obtained for graphs requiring $1,200$ physical qubits.
However, the uncertainty in the scaling behaviour means there is huge uncertainty in this figure, with relatively minor changes in the parameters meaning that any estimated point of ``hybrid equality'' can vary by at least 50\% (the uncertainty is particularly large on the upper end of the scale).
Moreover, one should caution that the scaling may also change for larger problems; indeed, while the minimum annealing time of $t_\text{anneal}=20\mu$s was used for all problem instances here, for larger problems this is no longer likely to be optimal~\cite{Boixo:2014aa,King17}.
The consequent need to consider the scaling of $t_\text{anneal}$ in addition to $k_{99}$ is likely to change future scaling behaviour, as are developments and improvements in future devices (e.g.\ by decreasing errors arising from noise and limits on the control of qubits).
Such an estimate for hybrid equality should thus be taken extremely cautiously, at the very best as a crude lower bound on the size of problem that one must at least be capable of solving with a quantum annealer before any quantum advantage is obtainable, and without any guarantee that such a condition be sufficient.

While $1,200$ physical qubits is not far beyond the size of D-Wave device we used (and within the reach of more recent devices), the fact that we had to reject many graphs requiring many fewer qubits because the quantum annealer could not find the optimal solution shows that the number of physical qubits itself is not necessarily the only limiting factor in this respect.
It is also worth noting that, even in Figure~\ref{fig:scalingOverall}(b), there is significant variation between different types of graphs and, indeed, certain graph families.
It is thus interesting to also look at the scaling behaviour for different graph families individually, as one may then make more informed estimates of when an advantage may be obtained on such graphs even if such families are not representative of arbitrary problems (both for the quantum and classical algorithms).
In Figure~\ref{fig:scalingFamilies} we show this for the cycle graphs $C_n$, star graphs $S_n$ and the complete graphs $K_n$ (each plotted as a function of $n$); as the $K_{n,m}$ graphs show much greater variability and have two parameters we avoid analysing them further here.

\begin{figure}[t]
\begin{center}
\begin{tabular}{ccc}
\hspace{-3mm}\includegraphics[width=0.36\textwidth]{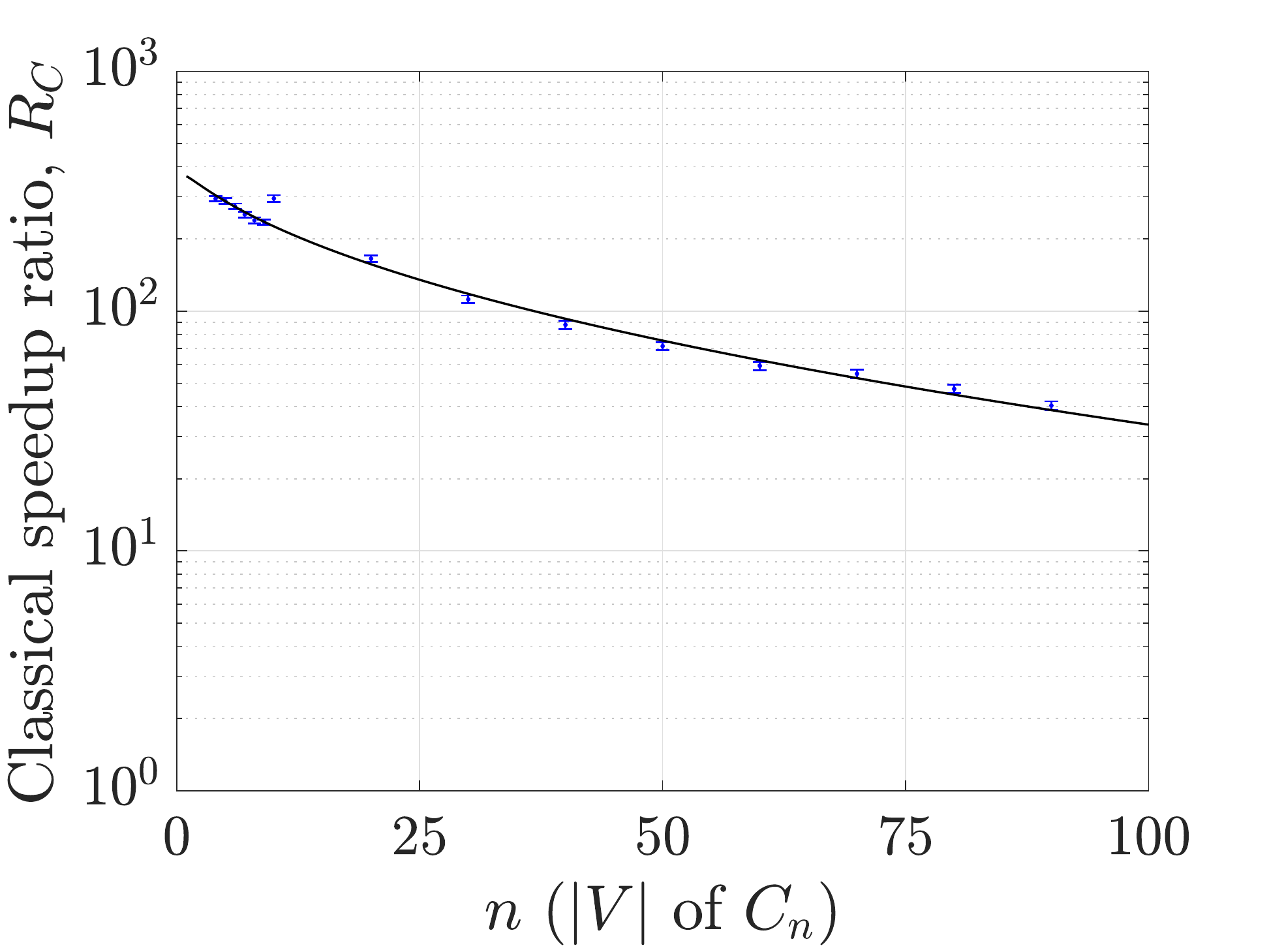} &\hspace{-10mm} \includegraphics[width=0.36\textwidth]{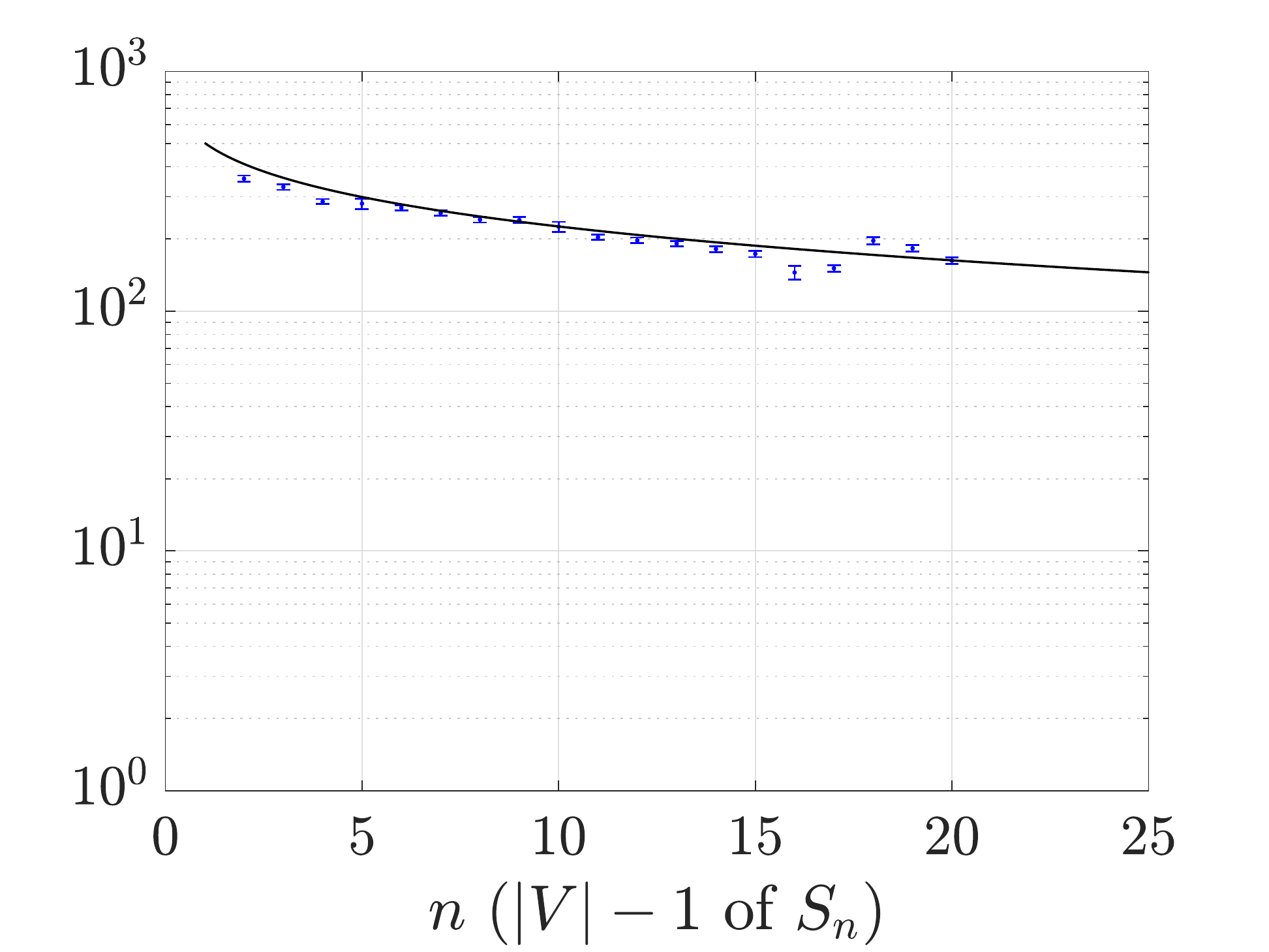} & \hspace{-10mm} \includegraphics[width=0.36\textwidth]{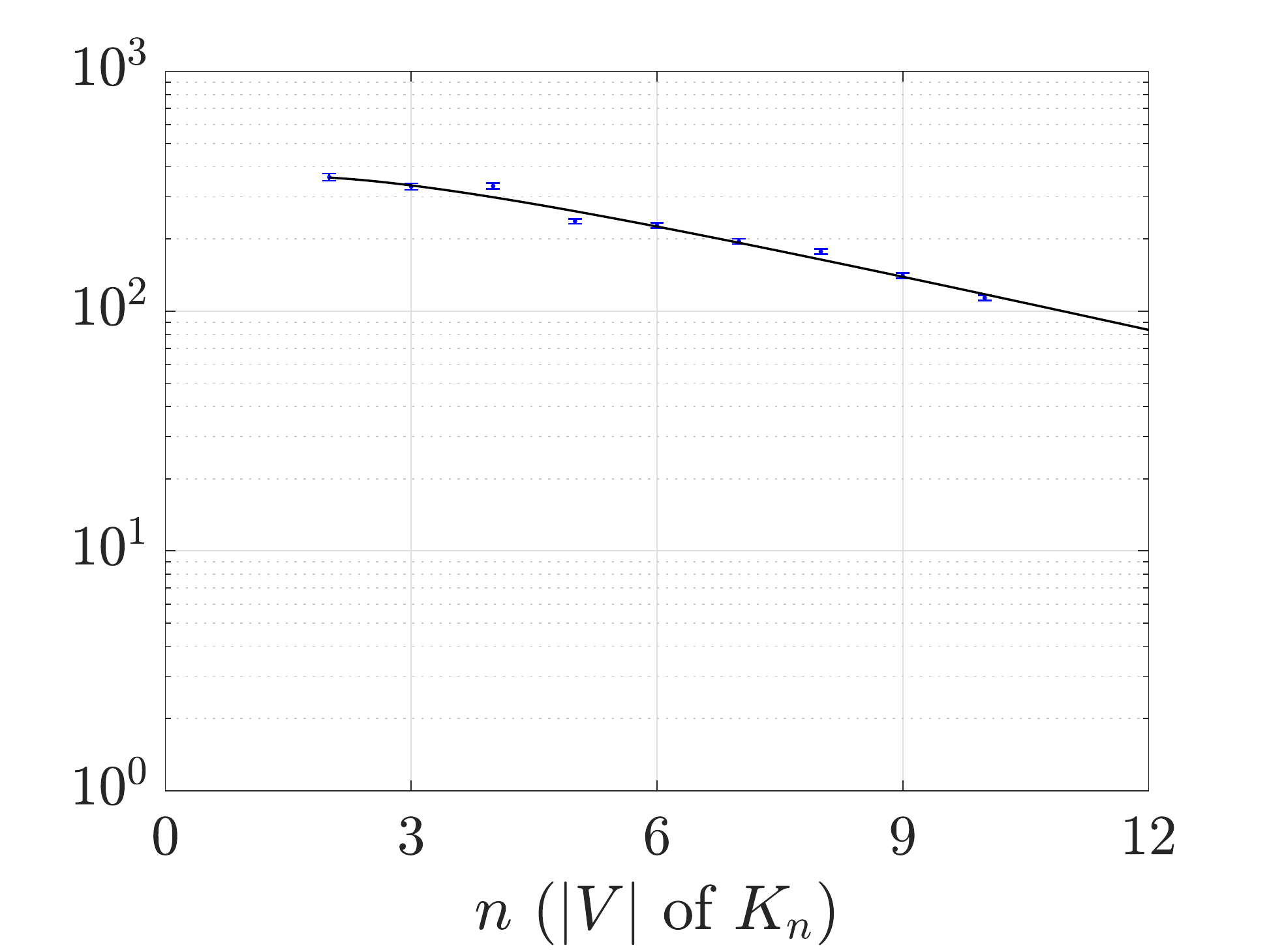}\\
(a) & (b) & (c)
\end{tabular}
\caption{Plots of the classical speedup ratio $R_C$ against $n$ for three families of graphs parameterised by $n$: (a) the $C_n$ graphs; (b) the $S_n$ graphs; (c) the $K_n$ graphs.}
\label{fig:scalingFamilies}
\end{center}
\end{figure}

Again the scaling behaviour is found to be consistent with a ratio of exponentials, but with much less uncertainty (note that, nonetheless, the log-scale used in Figure~\ref{fig:scalingFamilies} makes the uncertainty look smaller that it remains).
From these fits, we can extrapolate to estimate lower bounds on the point of ``hybrid equality'' (i.e., when $R_C=1$) for these three families as being obtained for $C_{580}$, $S_{5618}$ and $K_{38}$, respectively.
This provides a useful, albeit rough, estimate of when our algorithm might show a hybrid speedup on graphs taken from these families.
A necessary requirement is thus the ability for the heuristic embedding algorithm used in our hybrid algorithm to embed such logical graphs in the physical graph.

Of course, for such families one can generally devise analytic approaches to provide much smaller embeddings than the heuristic approach would find.
Indeed, cycle graphs permit small embeddings and $C_{580}$ can be embedded in the Chimera graph $\chi_{10}$ with $800$ physical qubits%
\footnote{A simple argument shows that there exists at cycle of length at least $\frac{7}{8}|\chi_{n}|$ by finding a cycle connecting
the bipartite blocks, where at least 7 of 8 vertices of each $K_{4,4}$ are spliced into a bigger cycle.}
while, as mentioned earlier in Section~\ref{sec:Chimera}, $K_{38}$ can also be embedded in $\chi_{10}$ and
$S_{5618}$ would require a much larger $\chi_{31}$ graph.%
\footnote{Another argument shows that we can construct in $\chi_{n}$ a spanning caterpillar with $2n^2$ spine vertices with $6n^2$
leaves.  Contracting the spine vertices. gives a minor embedding of $S_{6n^2}$.}
However, we emphasise that our algorithm is necessarily general and must thus be applicable to arbitrary problems.
Indeed, it is important to note that, if one were to tailor the algorithm for specific graph families, then much more efficient classical algorithms can easily be found.
For example, the MWIS of a $K_n$ graph can simply be computed as $\max_{v\in V}w(v)$ since the only independent sets are
singletons.
More generally, for families with low tree-width (which includes both cycle and star graphs), efficient algorithms are well known~\cite{kask05}.

In practice, one would thus need somewhat larger Chimera graphs to embed these graphs than the figures above suggest.
Nonetheless, they provide useful lower bounds on the size of a quantum annealer required to embed the problems for which hybrid equality might be expected to be obtained.
Since the D-Wave 2X device we used has a $\chi_{12}$ physical graph (albeit with some physical qubits disabled), one is not far from being able to embed the $K_{38}$ and $C_{580}$, and indeed this is probably feasible with newer devices.
However, our results showed that, at least for certain families of graphs, the prohibitory factor to obtaining a potential quantum speedup is not the number of physical qubits, but the stability and control one has over those qubits.
This is pointedly highlighted by noting that many problems that are easily embeddable in   D-Wave 2X's physical graph nonetheless fail to be solved by it~\cite{maxindepsetqubo2017} and, indeed, the larger problems in the graph families we solved were on the edge of what we could solve given the physical resources and time available to us.
The precision with which parameters can be controlled may play a major role in this~\cite{Pudenz:2016aa,albash19,PhysRevA.88.062314} and mitigating this will be a major challenge in the search for a practical quantum speedup.

Given discussions above, these estimates should only be seen as very conservative lower bounds for when a hybrid speedup may become obtainable, at least for some problem instances: not only may the scaling behaviour change for larger problem instances, but one should also recall that a speedup over a particular classical algorithm---here
the BIP-based solver---only proves a potential quantum speedup.
Indeed, as we noted above, for certain graph families very efficient solutions exist, while one would expect more efficient classical algorithms for the DWMIS problem to exist.
Nonetheless our results show that a ``potential'' quantum speedup remains plausible in the future for the DWMWIS problem, even if it is currently beyond the capabilities of the D-Wave annealer.

While our results  failed to find a quantum speedup and produced only tentative
evidence that such a speedup might be obtainable in the future for the DWMWIS
problem, the experiment was a successful proof-of-concept for the hybrid
paradigm we have presented.  In particular, the hybrid algorithm we presented
provided large absolute gains over the standard quantum approach and showed
good scaling behaviour.  As larger and more efficient devices become available
and more problems of practical interest are studied, it will become clearer if/when
a quantum speedup might be obtainable in practise.
A more detailed study using a TTT metric might allow larger problems to be studied (since one need not find an exact solution to each problem instance) and thereby lead to a better understanding in this direction.
However, this goes beyond the scope of the simple experiment we performed here and represents an interesting challenge for future research.

\section{Conclusion}
\label{sec:conclusions}

In this paper, we presented a hybrid quantum-classical paradigm for exploiting raw quantum speedups in quantum annealers.
Our paradigm is relevant in particular for devices in
which physical qubits have limited connectivity, where a problem of interest
must be embedded into the graph this connectivity imposes. This problem is a
major, but often neglected, hurdle to practical quantum computing.  Indeed, not
only does the need to find such an embedding often contribute significantly to the
overall computational costs, but the quality or size of embedding used can
often significantly affect the performance and accuracy of the quantum algorithm
itself~\cite{Lechner:2015aa,Vinci:2015aa}.

The paradigm we presented is not simply an algorithmic approach, but also aims to identify types of problems that are more amenable to quantum annealing.
In particular, we identify those problems that require solving a large number of
related subproblems, each of which can be directed solved via annealing, may
permit a hybrid approach. This is obtained by reusing and modifying embeddings for the related
subproblems.
Previous applications of quantum annealers have focused on problems that are not easily subdivided in this way, so even when only very simple reuse of embeddings is required---as in the case study we presented---the realisation that quantum annealing may be more advantageous for such problems is already important.
One can, however, envisage problems where the reuse of embeddings is more involved, such as small perturbations to the logical graph~\cite{Harary97,Goyal18}.
More research is needed to identify such problems of interest where the hybrid paradigm is applicable.

To exemplify the hybrid approach in an experimental setting, we identified a simple but suitable problem, called the dynamically-weighted maximum
weight independent set problem.  We experimentally solved a large number of
such instances on a D-Wave 2X quantum annealer, and observed the expected
advantage of the hybrid algorithm over a more traditional approach in which a known
embedding is not reused.  We failed to observe a quantum speedup over
classical algorithms, although this was not the main goal of the proof-of-concept
experiment. This is perhaps unsurprising given that many examples of quantum
annealing competing well with classical algorithms are on problems specifically constructed so that embedding is not an
issue~\cite{Boixo:2014aa,Denchev:2016aa,Hen:2015aa,Ronnow:2014aa,Shin:2014aa}. 
We note that another recent experimental study of the (unweighted) maximum independent sets problem conducted on the D-Wave 2000Q machine (the generation following the   D-Wave 2X device we utilised, for which the number of qubits has been doubled), was similarly restricted to graphs with no more than 70 vertices and also failed to observe a speedup~\cite{maxindepsetqubo2017}; in principle, the weighted version of the problem should be even harder for D-Wave devices because of analogue programming errors and the extra constraints the weights impose.
Nonetheless, our hybrid algorithm showed good scaling behaviour, providing
tentative evidence that a quantum speedup might be obtainable in the future.

While the problem we implemented as a proof-of-principle is perhaps somewhat contrived, it illustrates the advantage and feasibility of our hybrid approach and sets the groundwork for addressing more complex problems of practical interest.
Finding such problems, in which the same embedding can be reused multiple times, is itself a major step towards finding practical uses for quantum computers in the near term future.
One possible, more realistic, such problem is the decoding of error correcting codes~\cite{chancellor16}, and studying such problems would be an interesting next step towards obtaining quantum speedups from hybrid approaches.

\section*{Acknowledgements}
We thank  N.~Allen, C.~McGeoch, K.~Pudenz and S.~Reinhardt for fruitful discussions and critical comments.
This work has  been supported in part by the  Quantum Computing Research Initiatives at Lockheed Martin. 
We also thank the referees for detailed reports which improved the content and presentation of the paper.

\nocite{apsrev41Control} 
\bibliographystyle{apsrev4-1_modified}
\bibliography{LHM}

\appendix

\newpage
\section{Summary of results for MWDWIS instances}\label{app:resultSummary}

All the standard graphs were produced using SageMath~\cite{sagemath} and descriptions of them can be found in the corresponding API; the sole exception is the Dinneen Graph, which is described in Ref.~\cite{DeTemple93}.

\makeatletter
\def\dual#1{\expandafter\dual@aux#1\@nil}
\def\dual@aux#1/#2\@nil{\begin{tabular}{@{}c@{}}#1\\\raisebox{4pt}{#2}\end{tabular}}
\makeatother

\begin{table}[h!]
	\begin{center}
		\caption{Table summarising the 124 graphs defining the DWMWIS problem instances and the average times for the hybrid algorithm, the classical BIP-based algorithm, and the standard quantum annealing approach.}\label{table:resultSummary}
	\end{center}
\end{table}
\vspace{-5mm}
\begin{center}
		\begin{longtable}{lccccccc}
			\hline
			\hline
				Graph $G=(V,E)$		 & $|V|$ & $|E|$ & \dual{max/chain} & \dual{embedded/graph order}   &  $T_H$ (ms)   &    $T_C$ (ms)   &   $T_\text{std}$ (ms)\\
				\hline
			    Bidiakis Cube        & 12    &    18 &2 &18   & $4635\pm{102}$   &   $26.6\pm 0.4$    &   $22851\pm 184$\\
			    Blanusa Snark 1  & 18    &    27 &3 &33   & $5799\pm{120}$   &   $39.2\pm 0.8$    &   $28846\pm 591$  \\  
			    Blanusa Snark 2 & 18    &    27 &2 &31   & $6280\pm{139}$   &   $38.9\pm 0.7$    &   $28802\pm 405$    \\
			    Brinkmann            & 21    &    42 &4 &68   &  $12988^{+861}_{-363}$   &   $66.8\pm 0.6$    &   $42876^{+1047}_{-698}$    \\
			    Bucky Ball           & 60    &    90 &3 &127   &  $12491^{+599}_{-286}$   &   $123.1\pm 3.3$    &   $83128^{+4462}_{-4431}$    \\
			    Bull                 &  5    &     5 &2 &6   & $4379\pm{90}$   &   $16.4\pm 0.3$    &   $18427\pm 99$    \\
			    Butterfly            &  5    &     6 &2 &6   & $4405\pm{91}$   &   $17.3\pm 0.3$    &   $19137\pm 99$    \\
				$C_4$                &  4    &     4 &1 &4   &   $4441\pm{89}$   &   $15.1\pm 0.2$    &   $19162\pm 447$    \\
				$C_5$                &  5    &     5  &2 & 6  & $4785\pm{109}$   &   $16.6\pm 0.3$    &   $19209\pm 133$    \\
				$C_6$                &  6    &     6  &1 & 6  & $4781\pm{103}$   &   $17.5\pm 0.3$    &   $19532\pm 140$    \\
				$C_7$                &  7    &     7  &2 & 8  &   $4785\pm{102}$   &   $18.9\pm 0.4$    &   $20110\pm 176$    \\
				$C_8$                &  8    &     8  &2 & 10  & $4743\pm{102}$   &   $19.8\pm 0.4$    &   $20375\pm 174$    \\
				$C_9$                &  9    &     9  &2 & 10  & $4927\pm{107}$   &   $20.9\pm 0.3$    &   $21084\pm 149$    \\
			    $C_{10}$             & 10    &    10  &2 & 12  & $6453\pm{161}$   &   $21.9\pm 0.5$    &   $22877\pm 194$    \\
			    $C_{20}$             & 20    &    20  &3 &40  & $5788\pm{142}$   &   $35.0\pm 0.7$    &   $28330\pm 640$    \\
			    $C_{30}$             & 30    &    30  &2 &44  &   $5436\pm{135}$   &    $48.5\pm 1.3$    &   $33394\pm 512$    \\
			    $C_{40}$             & 40    &    40  &3 & 71  & $5490\pm{123}$   &    $62.6\pm 2.0$    &   $41743\pm 1043$    \\
			    $C_{50}$             & 50    &    50  &2 &92  & $5644\pm{123}$   &    $78.9\pm 2.2$    &   $50867\pm 1190$    \\
			    $C_{60}$             & 60    &    60  &3 &106  & $5560\pm{120}$   &   $94.1\pm 3.3$    &   $58397\pm 2378$     \\  
			    $C_{70}$             & 70    &    70  &2 & 92  & $6122\pm{117}$   &   $111.8\pm 4.3$    &   $70066\pm 2245$    \\    
			    $C_{80}$             & 80    &    80  &4 &142  & $6084\pm{123}$   &   $128.4\pm 4.7$    &   $79117\pm 3279$    \\  
			    $C_{90}$             & 90    &    90 &2 &128   & $6006\pm{120}$   &   $148.8\pm 5.6$    &   $98769\pm 4681$    \\    
			    Chvatal              & 12    &    24 &3 &25   & $5899^{+124}_{-122}$   &   $35.4\pm 0.4$    &   $26372^{+439}_{-438}$    \\
			    Clebsch              & 16    &    40 &4 &50   & $8527^{+172}_{-160}$   &   $60.2\pm 0.6$    &   $35207^{+818}_{-816}$    \\
			    Coxeter              & 28    &    42&3 &57    & $8424^{+205}_{-181}$   &   $53.9\pm 1.3$    &   $39807^{+575}_{-567}$    \\
			    Desargues            & 20    &    30&2 &28    & $6160^{+126}_{-124}$   &   $37.3\pm 0.7$    &   $30861\pm 672$    \\
			    Diamond              &  4    &     5 &2 &5   &   $4783\pm{106}$   &   $16.0\pm 0.2$    &   $19089\pm 111$    \\
			    Dinneen              &  9    &    21 &3 &22   & $6072\pm{126}$   &     $29.6\pm 0.6$    &   $24724\pm 285$    \\
			    Dodecahedral         & 20    &    30 &3 &39   & $6128^{+124}_{-122}$   &   $45.6\pm 0.9$    &   $31373^{+997}_{-996}$    \\
			    Double Star Snark    & 30    &    45 & 3 &65   & $8527^{+214}_{-192}$   &   $58.0\pm 1.3$    &   $40801^{+773}_{-767}$    \\
			    Durer                & 12    &    18  &2 &19  & $4643\pm{100}$   &   $30.3\pm 0.3$    &   $23076\pm 254$    \\
			    Dyck                 & 32    &    48  &3 &73  &  $10562^{+673}_{-275}$   &   $55.6\pm 1.6$    &   $44380^{+1185}_{-1013}$    \\
			    Ellingham Horton  & 54    &    81 &3 &117   &   $8043^{+232}_{-152}$   &   $93.5\pm 3.0$    &   $63265^{+2007}_{-1999}$    \\
			    Errera               & 17    &    45 &4 &53    & $9543^{+201}_{-182}$   &   $90.1\pm 0.9$    &   $39738^{+867}_{-863}$    \\
			    Flower Snark         & 20    &    30 &2 &35   & $5589\pm{105}$   &    $39.5\pm 0.7$    &   $28992\pm 341$    \\
			    Folkman              & 20    &    40 &4 & 49   &  $10293^{+471}_{-258}$   &   $39.5\pm 0.7$    &   $38964^{+853}_{-757}$    \\
			    Franklin             & 12    &    18 &2 &21   & $5030\pm{99}$   &   $25.3\pm 0.4$    &   $23127\pm 165$    \\
			    Frucht               & 12    &    18 &2 & 18   & $4842\pm{101}$   &   $29.5\pm 0.5$    &   $23791\pm 349$    \\
			    Goldner Harary       & 11    &    27 &4 &22   & $5716^{+132}_{-119}$   &   $28.2\pm 0.4$    &   $26486^{+381}_{-377}$    \\
			    $2\times 3$ Grid     &  6    &     7 &1 &6   & $5073\pm{140}$   &   $17.6\pm 0.2$    &   $19972\pm 162$    \\
			    $3\times 3$ Grid     &  9    &    12 &2 &11   &   $5336\pm{150}$   &   $21.1\pm 0.3$    &   $21948\pm 258$    \\
			    $3\times 4$ Grid     & 12    &    17 &3 & 19   & $5122\pm{107}$   &   $25.0\pm 0.4$    &   $24100\pm 447$    \\
			    $4\times 4$ Grid     & 16    &    24  &2 &25  & $5409\pm{140}$   &   $31.7\pm 0.6$    &   $27605\pm 551$    \\
			    $4\times 5$ Grid     & 20    &    31  &3 & 35  & $6999^{+155}_{-153}$   &   $37.2\pm 1.0$    &   $32956\pm 693$    \\
			    $6\times 6$ Grid     & 36    &    60  &3 &72  & $7743^{+195}_{-184}$   &   $65.7\pm 3.3$    &   $54679^{+2383}_{-2382}$    \\
			    $6\times 7$ Grid     & 42    &    71  &4 & 11  &  $10252^{+1122}_{-287}$   &   $76.6\pm 4.4$    &   $64583^{+2739}_{-2516}$    \\
			    $7\times 7$ Grid     & 49    &    84  &4 & 120  & $8591^{+213}_{-183}$   &    $85.2\pm 2.3$    &   $75158^{+4197}_{-4195}$    \\
			    Grotzsch             & 11    &    20  &3 &22  & $5793\pm{133}$   &   $29.7\pm 0.3$    &   $24741\pm 324$    \\
			    Heawood              & 14    &    21  & 3 &25  & $7663^{+197}_{-193}$   &   $29.0\pm 0.6$    &   $27542^{+380}_{-379}$    \\
			    Herschel             & 11    &    18  &3 &25  & $5871\pm{145}$   &   $24.3\pm 0.3$    &   $24394\pm 349$    \\
			    Hexahedral           &  8    &    12  &1 &8  & $4803\pm{106}$   &   $20.6\pm 0.3$    &   $20920\pm 145$    \\
			    Hoffman              & 16    &    32  &4 &38  & $7010^{+168}_{-167}$   &    $33.2\pm 0.6$    &   $29453\pm 433$    \\
			    House                &  5    &     6  &2 &6  & $4700\pm{110}$   &   $16.9\pm 0.3$    &   $19292\pm 113$    \\
			    Icosahedral          & 12    &    30  &4 &39  & $7177^{+138}_{-125}$   &   $50.0\pm 0.4$    &   $29413^{+422}_{-418}$    \\
				$K_2$                &  2    &  1 &   1 &2    & $4607\pm{109}$   &   $12.7\pm 0.3$    &  $4607\pm 109$    \\
				$K_3$                &  3    &     3 &2 &4    & $4821\pm{118}$   &   $14.6\pm 0.3$    &  $4821\pm 118$    \\
				$K_4$                &  4    &     6 & 2 & 6   &   $5875\pm{131}$   &   $17.7\pm 0.3$    &    $5875\pm 131$    \\
				$K_5$                &  5    &    10 &2 &8   & $5210\pm{119}$   &   $22.0\pm 0.2$    &  $5210\pm 119$    \\
				$K_6$                &  6    &    15 &3 &14   & $6101\pm{143}$   &   $26.8\pm 0.3$    &  $6101\pm 143$    \\
				$K_7$                &  7    &    21  &3 & 18  & $6546^{+158}_{-157}$   &    $33.5\pm 0.2$    &   $27296\pm 2667$    \\
				$K_8$                &  8    &    28  &4 &23  & $7293\pm{180}$   &   $41.2\pm 0.4$    &   $28836\pm 290$    \\
				$K_9$                &  9    &    36 &4 &27    &   $6883\pm{164}$   &   $49.0\pm 0.5$    &   $30247\pm 457$    \\
			    $K_{10}$             & 10    &    45 & 5 &36   & $6726^{+153}_{-148}$   &   $59.3\pm 0.6$    &   $33090^{+856}_{-855}$    \\
			    $K_{2,3}$            &  5    &     6  &1 &5  & $5570\pm{142}$   &   $16.3\pm 0.2$    &   $19083\pm 152$    \\	    
			    $K_{3,3}$            &  6    &     9  &1 &6  & $4486\pm{103}$   &   $17.7\pm 0.3$    &  $4486\pm 103$    \\
			    $K_{3,4}$            &  7    &    12  &1 & 7  & $5147\pm{125}$   &   $19.3\pm 0.3$    &   $19641\pm 487$    \\		    
			    $K_{4,4}$            &  8    &    16  &1 & 8  & $5036\pm{123}$   &   $21.4\pm 0.3$    &  $5036\pm 123$    \\
			    $K_{4,5}$            &  9    &    20    & 2 &13 & $5729\pm{131}$   &   $23.6\pm 0.3$    &   $20173\pm 136$    \\		    
			    $K_{5,5}$            & 10    &    25 &2 &20   & $7470\pm{215}$   &   $25.4\pm 0.3$    &  $7469\pm 215$    \\
			    $K_{5,6}$            & 11    &    30  &2 &22  & $8619\pm{212}$   &   $26.8\pm 0.3$    &   $23805\pm 216$    \\
			    $K_{5,7}$            & 12    &    35  &2 &24  & $6563\pm{155}$   &   $28.7\pm 0.4$    &   $28026\pm 292$    \\
			    $K_{5,8}$            & 13    &    40  &2 &26  & $4789\pm{74}$   &   $30.4\pm 0.4$    &   $27103\pm 214$    \\
			    $K_{5,9}$            & 14    &    45  &3 &33  & $6705^{+154}_{-151}$   &   $30.9\pm 0.5$    &   $31346^{+352}_{-351}$    \\
			    $K_{6,6}$           & 12    &    36 &2 &24   &   $6992.0\pm{159}$   &   $28.9\pm 0.3$    &   $23674\pm 231$    \\
			    $K_{6,7}$            & 13    &    42 &2 &26   & $6279.8\pm{125}$   &   $31.4\pm 0.4$    &   $30079\pm 305$    \\
			    $K_{6,8}$            & 14    &    48 &2 &28   & $6353.1\pm{131}$   &   $33.3\pm 0.5$    &   $32331\pm 539$    \\
			    $K_{6,9}$            & 15    &    54 &3 &36   & 7$089^{+192}_{-168}$   &    $33.9\pm 0.5$    &   $38878^{+5248}_{-5247}$    \\
			    $K_{7,7}$            & 14    &    49 &2 &28   & $6480\pm{132}$   &   $33.5\pm 0.4$    &   $32279\pm 941$    \\
			    $K_{7,8}$            & 15    &    56 &2 &30   & $6563\pm{154}$   &   $35.8\pm 0.5$    &   $33432\pm 599$    \\
			    $K_{8,8}$            & 16    &    64  & 2 &32  &   $6319\pm{150}$   &   $38.4\pm 0.6$    &   $34722\pm 761$    \\
			    $K_{8,9}$            & 17    &    72  &3 &51  &   $6416^{+145}_{-137}$   &   $40.8\pm 0.8$    &   $44115\pm 5996$    \\
			    $K_{9,9}$            & 18    &    81  &3 &54  & $6424\pm{134}$   &   $44.1\pm 0.6$    &   $40895\pm 712$    \\
			    $K_{10,10}$          & 20    &   100  &3 &60  & $5711\pm{109}$   &   $50.0\pm 1.0$    &   $47113\pm 1408$    \\
			    $K_{11,11}$          & 22    &   121  &3 &66  & $6782^{+134}_{-130}$   &   $57.4\pm 1.1$    &   $53698\pm 1458$    \\
			    $K_{12,12}$          & 24    &   144  &4 &96  &  $33536^{+1674}_{-852}$   &   $67.2\pm 0.7$    &   $86818^{+2241}_{-1717}$    \\
			    Kittell              & 23    &    63  &5 & 83  &  $11920^{+427}_{-217}$   &   $177.8\pm 8.2$    &   $52401^{+1959}_{-1924}$    \\
			    Krackhardt Kite      & 10    &    18  &3 &15  & $5048\pm{99}$   &   $29.3\pm 0.4$    &   $22155\pm 263$    \\
			    Markstroem           & 24    &    36 &9 &87   & $5547\pm{130}$   &    $59.0\pm 1.2$    &   $32525\pm 568$    \\
			    McGee                & 24    &    36  &3 &44  & $7309^{+155}_{-148}$   &   $48.6\pm 1.2$    &   $35504^{+1003}_{-1002}$    \\
			    Moebius Kantor       & 16    &    24  &2 &24  &   $6420\pm{155}$   &   $31.4\pm 0.6$    &   $27170\pm 361$    \\
			    Moser Spindle        &  7    &    11  &2 &10  & $5326\pm{131}$   &  $23.4\pm 0.4$    &   $21473\pm 241$    \\
			    Nauru                & 24    &    36  &3 &55  & $7862^{+180}_{-171}$   &   $43.2\pm 1.1$    &   $34622^{+702}_{-700}$    \\
			    Octahedral           &  6    &    12  &2 &8  & $5461\pm{133}$   &   $21.1\pm 0.3$    &   $21262\pm 219$    \\
			    Pappus               & 18    &    27  &2 &32  & $6618\pm{179}$   &   $34.1\pm 0.7$    &   $28259\pm 398$    \\
			    Petersen             & 10    &    15  &3 &22  & $5069\pm{108}$   &    $24.5\pm 0.4$    &   $22275\pm 183$    \\
			    Poussin              & 15    &    39  &4 &46  & $8621^{+195}_{-182}$   &   $64.6\pm 0.8$    &   $35846^{+529}_{-525}$    \\
			    $Q_3$                &  8    &    12 &1 &8   & $5153\pm{99}$   &   $20.7\pm 0.2$    &   $22597\pm 1180$    \\
			    $Q_4$                & 16    &    32 &3 &36   &   $6091\pm{121}$   &   $33.0\pm 0.6$    &   $28643\pm 391$    \\
			    Robertson            & 19    &    38 &4 &55   & $9635^{+220}_{-187}$   &   $59.9\pm 0.5$    &   $36633^{+764}_{-755}$    \\
			    $S_2$                &  3    &     2 &1 &3   & $4858\pm{127}$   &   $13.6\pm 0.2$    &   $18580\pm 147$    \\
			    $S_3$                &  4    &     3 &1 &4   & $4849\pm{105}$   &   $14.7\pm 0.2$    &   $18738\pm 171$    \\
			    $S_4$                &  5    &     4 &1 &5   & $4506\pm{85}$   &   $15.7\pm 0.3$    &   $18406\pm 93$    \\
			    $S_5$                &  6    &     5 &2 &7   & $4977\pm{103}$   &   $17.7\pm 0.8$    &   $19204\pm 178$    \\
			    $S_6$                &  7    &     6  &1 &7  & $4766\pm{102}$   &   $17.7\pm 0.3$    &   $20319\pm 899$    \\
			    $S_7$                &  8    &     7 &2 &9   & $4819\pm{98}$   &   $18.8\pm 0.3$    &   $22570\pm 1238$    \\
			    $S_8$                &  9    &     8  &3 &11  &   $4807\pm{94}$   &   $20.0\pm 0.4$    &   $20251\pm 225$    \\
			    $S_9$                & 10    &     9  &3 &12  & $4994\pm{125}$   &   $20.9\pm 0.3$    &   $20042\pm 159$    \\
			    $S_{10}$             & 11    &    10  &4 &14  & $5290\pm{156}$   &     $23.5\pm 0.9$    &   $20457\pm 222$    \\
			    $S_{11}$             & 12    &    11  &4 &15  & $4738\pm{92}$   &   $23.3\pm 0.4$    &   $23587\pm 3131$    \\
			    $S_{12}$             & 13    &    12  &5 &17  & $4814\pm{100}$   &   $24.4\pm 0.4$    &   $21258\pm 281$    \\
			    $S_{13}$             & 14    &    13  &3 &16  & $4896\pm{98}$   &   $25.6\pm 0.4$    &   $21003\pm 300$    \\
			    $S_{14}$             & 15    &    14  &4 &18  & $4772\pm{90}$   &     $26.3\pm 0.6$    &   $20860\pm 211$    \\
			    $S_{15}$             & 16    &    15  &4 &19  & $4738\pm{104}$   &   $27.3\pm 0.6$    &   $21627\pm 270$    \\
			    $S_{16}$             & 17    &    16  &4 &20  & $4432\pm{84}$   &   $30.6\pm 1.9$    &   $21143\pm 216$    \\
			    $S_{17}$             & 18    &    17 &6 &23   & $4444\pm{84}$   &   $29.5\pm 0.8$    &   $22650\pm 361$    \\
			    $S_{18}$             & 19    &    18  &8 &26  & $6113\pm{122}$   &    $31.1\pm 0.8$    &   $24339\pm 471$    \\
			    $S_{19}$             & 20    &    19  &6&25  & $6020\pm{123}$   &   $32.9\pm 0.7$    &   $23474\pm 443$    \\
			    $S_{20}$             & 21    &    20  &7 &27  & $5569\pm{121}$   &   $34.3\pm 0.8$    &   $24497\pm 619$    \\
			    Shrikhande           & 16    &    48  &5 &70  &  $12803^{+367}_{-244}$   &   $86.1\pm 0.7$    &   $45275^{+1106}_{-1072}$    \\
			    Sousselier           & 16    &    27  &3 &31  & $7231^{+171}_{-169}$   &   $38.9\pm 0.8$    &   $29675^{+531}_{-530}$    \\
			    Thomsen              &  6    &     9  &1 &6  & $5220\pm{137}$   &   $17.7\pm 0.2$    &   $20555\pm 190$    \\
			    Tietze               & 12    &    18  &3 &23  & $4927\pm{113}$   &   $27.6\pm 0.3$    &   $23380\pm 216$    \\
			    TutteCoxeter         & 30    &    45  &3 &55  & $8566^{+191}_{-179}$   &   $52.5\pm 1.2$    &   $40058^{+594}_{-591}$    \\
			    Wagner               &  8    &    12  &2 &16  & $4817\pm{111}$   &   $22.0\pm 0.3$    &   $21004\pm 191$	  \\
				\hline
				\hline
		\end{longtable}
	\end{center}

\begin{table}[h!]
\caption{A list of 32 small graphs for which a full analysis could not be performed due to the quantum annealer not finding enough optimal solutions on each of the 100 MWIS instances.}\label{table:resultSummary2}

\medskip

\begin{center}
\begin{tabular}{lcccc} \hline
				\hline
Graph $G=(V,E)$		 & $|V|$ & $|E|$    &  \dual{max/chain} & \dual{embedded/graph order} \\ \hline
Balaban 10-Cage & 70 & 105 & 5 & 232 \\
BiggsSmith & 102 & 153 & 8 & 358 \\
Ellingham Horton 2 & 78 & 117 & 3 & 172 \\
Foster & 90 & 135 & 5 & 285 \\
Gray & 54 & 81 & 4 & 148 \\
$6\times 8$ Grid & 48 & 82 & 4 & 118 \\
$6\times 9$ Grid & 54 & 93 & 5 & 134 \\
$7\times 8$ Grid & 56 & 97 & 4 & 141 \\
$7\times 9$ Grid & 63 & 110 & 4 & 149 \\
$8\times 8$ Grid & 64 & 112 & 4 & 151 \\
$8\times 9$ Grid & 72 & 127 & 5 & 201 \\
$9\times 9$ Grid & 81 & 144 & 4 & 205 \\
Harries & 70 & 105 & 5 & 201 \\
Harries Wong & 70 & 105 & 5 & 205 \\
Hoffman Singleton & 50 & 175 & 14 & 501 \\
Horton & 96 & 144 & 5 & 218 \\
$K_{11}$ & 11 & 55 & 5 & 50 \\
$K_{12}$ & 12 & 66 & 5 & 53 \\
$K_{13}$ & 13 & 78 & 5 & 56 \\
$K_{14}$ & 14 & 91 & 6 & 72 \\
$K_{15}$ & 15 & 105 & 7 & 85 \\
Ljubljana & 112 & 168 & 7 & 390 \\
Meredith & 70 & 140 & 7 & 210 \\
$Q_5$ & 32 & 80 & 6 & 128 \\
$Q_6$ & 64 & 192 & 11 & 429 \\
Schlaefli & 27 & 216 & 11 & 234 \\
SimsGewirtz & 56 & 280 & 25 & 812 \\
Sylvester & 36 & 90 & 8 & 177 \\
Szekeres Snark & 50 & 75 & 6 & 129 \\
Tutte 12-Cage & 126 & 189 & 10 & 529 \\
Wells & 32 & 80 & 8 & 138 \\
Wiener Araya & 42 & 67 & 4 & 99 \\
\hline \hline
\end{tabular}
\end{center}
\end{table}

\end{document}